\documentclass[letterpaper,11pt]{article}
\usepackage[margin=1in]{geometry}
\usepackage[bookmarks, colorlinks=true, plainpages = false, colorlinks=true,
            citecolor=red,
            linkcolor=blue,
            anchorcolor=red,
            urlcolor=blue]{hyperref}
\usepackage{url}
\usepackage{amsmath,amsfonts,amsthm,amssymb,bm, verbatim,dsfont,mathtools}
\usepackage{color,graphicx,appendix}
\usepackage{etoolbox}

\usepackage{array}
\usepackage{multirow}

\usepackage{float}
\usepackage{tikz}
\usetikzlibrary{matrix,arrows,calc,shapes,backgrounds}
\usetikzlibrary{shapes.callouts,decorations.text}

\usetikzlibrary{shapes.misc}

\tikzset{cross/.style={cross out, draw=black, minimum size=2*(#1-\pgflinewidth), inner sep=0pt, outer sep=0pt},
cross/.default={1pt}}

\tikzstyle{int}=[draw, fill=blue!20, minimum size=2em]
\tikzstyle{dot}=[circle, draw, fill=blue!20, minimum size=2em]
\tikzstyle{dotred}=[circle, draw, fill=red!20, minimum size=2em]
\tikzstyle{init} = [pin edge={to-,thin,black}]
\tikzstyle{initred} = [pin edge={to-,thin,red}]
\tikzstyle{plan}=[draw, fill=blue!20, minimum size=2em, text width=5em, rounded corners,align=center]
\tikzstyle{planwide}=[draw, fill=blue!20, minimum size=2em, text width=8em, rounded corners,align=center]
\tikzstyle{nodedot}=[circle, draw, fill=white, minimum size=0.3cm,inner sep=0pt]
\tikzstyle{nodedot}=[circle, draw, fill=white, minimum size=3,inner sep=0pt]
\tikzstyle{Medge}=[green!60!black, thick]
\tikzstyle{Bedge}=[red, thick]
\tikzstyle{Cedge}=[blue, thick]
\tikzstyle{Sedge}=[black, thick]
\tikzstyle{Mgiantedge}=[green!60!black, line width=3.0pt]
\tikzstyle{Bgiantedge}=[red,line width=3.0pt]
\tikzstyle{Cgiantedge}=[blue,line width=3.0pt]
\tikzstyle{Sgiantedge}=[black,line width=3.0pt]
\tikzstyle{shadedgiantnode}=[circle, draw, fill=black!10, minimum size=1cm, inner sep=0pt]
\tikzstyle{unshadedgiantnode}=[circle, draw, fill=white, minimum size=1cm, inner sep=0pt]
\makeatletter
\tikzset{my loop/.style =  {to path={
  \pgfextra{}
  [looseness=5,min distance=10mm]
  \tikz@to@curve@path},font=\sffamily\small
  }}  
\makeatletter 
\newcolumntype{C}[1]{>{\centering\arraybackslash}p{#1}}

\tikzstyle{vertexdot}=[circle, draw, fill=white, minimum size=3,inner sep=0pt]
\tikzstyle{root}=[circle, draw, fill=black, minimum size=3,inner sep=0pt]

\tikzstyle{vertexdotsolid}=[circle, draw, fill=black, minimum size=3,inner sep=0pt]

\usepackage{pgfplots}

\pgfplotsset{
    standard/.style={
        axis x line=middle,
        axis y line=middle,
        every axis x label/.style={at={(current axis.right of origin)},anchor=north west},
        every axis y label/.style={at={(current axis.above origin)},anchor=north west}
    }
}
\usepgfplotslibrary{fillbetween}

\usepackage{xr,xspace}
\usepackage{todonotes}
\usepackage{paralist}
\usepackage{enumitem}

\usepackage{caption,subcaption,soul}
\usepackage[ruled]{algorithm2e} 
\makeatletter
\usepackage{romanbar}
\usepackage{booktabs}
\usepackage{cleveref}
\usepackage{array}

\theoremstyle{plain}
\newtheorem{theorem}{Theorem}
\newtheorem{lemma}{Lemma}

\newtheorem{corollary}{Corollary}
\theoremstyle{definition}
\newtheorem{definition}{Definition}

\newtheorem{problem}{Problem}

\newtheorem*{remark*}{Remark}
\newtheorem*{theorem*}{Theorem}

\newcommand{\argmin}{\mathop{\arg\min}}

\usepackage{xspace,prettyref}

\newcommand{\reals}{{\mathbb{R}}}

\newcommand{\Expect}{\mathbb{E}}

\newcommand{\Prob}{\mathbb{P}}

\newcommand{\prob}[1]{ \mathbb{P}\left\{ #1 \right\} }

\newcommand{\Binom}{{\rm Binom}}

\newrefformat{eq}{(\ref{#1})}
\newrefformat{chap}{Chapter~\ref{#1}}
\newrefformat{sec}{Section~\ref{#1}}
\newrefformat{alg}{Algorithm~\ref{#1}}
\newrefformat{fig}{Figure~\ref{#1}}
\newrefformat{tab}{Table~\ref{#1}}
\newrefformat{rmk}{Remark~\ref{#1}}
\newrefformat{clm}{Claim~\ref{#1}}
\newrefformat{claim}{Claim~\ref{#1}}
\newrefformat{def}{Definition~\ref{#1}}
\newrefformat{cor}{Corollary~\ref{#1}}
\newrefformat{lmm}{Lemma~\ref{#1}}
\newrefformat{prop}{Proposition~\ref{#1}}
\newrefformat{prob}{Problem~\ref{#1}}
\newrefformat{app}{Appendix~\ref{#1}}
\newrefformat{hyp}{Hypothesis~\ref{#1}}
\newrefformat{thm}{Theorem~\ref{#1}}
\newrefformat{fac}{Fact~\ref{#1}}
\newrefformat{ex}{Example~\ref{#1}}

\newcommand{\norm}[1]{\left\|{#1} \right\|}

\newcommand{\calA}{{\mathcal{A}}}

\newcommand{\bbD}{{\mathbb{D}}}

\newcommand{\opt}{\mathrm{OPT}}
\newcommand{\optD}{\mathrm{OPT}_{\mathbb{D}}}
\newcommand{\cost}{\mathrm{cost}}
\newcommand{\costD}{\mathrm{cost}_{\bbD}}

\newcommand{\regret}{\mathrm{Regret}_{\mathbb{D}}}
\newcommand{\std}{\mathrm{std}}
\newcommand{\com}{\mathrm{com}}

\title{Online Metric Matching: Beyond the Worst Case\thanks{A preliminary version of this paper is published in WINE'24 with title \emph{Stochastic Online Metric Matching: Adversarial is no Harder than Stochastic}.}}

\author{Mingwei Yang\thanks{Stanford University. Email: \texttt{mwyang@stanford.edu}} \and Sophie H.\ Yu\thanks{The Wharton School of Business, University of Pennsylvania. Email: \texttt{hysophie@wharton.upenn.edu}.
}}

\date{\today}

\begin{document}
\maketitle

\begin{abstract}
    We study the online metric matching problem.
    There are $m$ servers and $n$ requests located in a metric space, where all servers are available upfront and requests arrive one at a time.
    Upon the arrival of a new request, it needs to be immediately and irrevocably matched to an available server, resulting in a cost of their distance.
    The objective is to minimize the total matching cost.
    
    When servers are adversarial and requests are independently drawn from a known distribution, we reduce the problem to a more tractable setting where servers and requests are all independently drawn from the same distribution.
    Applying our reduction, for $[0, 1]^d$ with various choices of distributions, we achieve improved competitive ratios and nearly optimal regret in both balanced and unbalanced markets.
    In particular, we give $O(1)$-competitive algorithms for $d \geq 3$ in both balanced and unbalanced markets with smooth distributions.
    Our algorithms improve on the $O((\log \log \log n)^2)$ competitive ratio of Gupta et al. (ICALP'19) for balanced markets in various regimes, and provide the first positive results for unbalanced markets.
    Moreover, when servers and requests are all adversarial, and a prediction of request locations is provided, we present a general framework for transforming an arbitrary algorithm that does not use predictions into an algorithm that leverages predictions.
    The transformation applies the given algorithm in a black-box manner, and the performance of the resulting algorithm degrades smoothly as the prediction accuracy deteriorates while preserving the worst-case guarantee.
\end{abstract}
\tableofcontents

\section{Introduction}

The development of centralized matching markets has emerged as a crucial research area across operational research, economics, and computer science, driven by its diverse applications in kidney exchange programs, ride-sharing services, and online advertising platforms. Market designers in these domains face the fundamental challenge of making matching decisions under future uncertainty, with the primary objective of maximizing the total incurred value. 
The growing significance of these applications has sparked extensive research across multiple related areas. These include online bipartite matching \cite{DBLP:conf/stoc/KarpVV90,DBLP:journals/fttcs/Mehta13,huang2024onlinematchingbriefsurvey}, online resource allocation \cite{DBLP:journals/mansci/VeraB21,besbes2023dynamic,DBLP:journals/ior/BalseiroLM23,DBLP:journals/ior/BalseiroBP24,bray2024logarithmic,jiang2024degeneracy}, and dynamic matching \cite{DBLP:journals/ior/AndersonAGK17,akbarpour2020thickness,DBLP:conf/sigmetrics/WeiXY23,gupta2024greedy,kerimov2025achieving}.

A particularly compelling example of these applications can be found in the recent proliferation of ride-hailing platforms such as Didi, Uber, and Lyft.
In these systems, the platform must make an immediate matching decision when a customer requests a ride---significantly delaying or rejecting the request is typically not an option. The primary cost metric in each match is the distance a driver needs to travel to reach their assigned customer.
While a simple greedy approach matches each arriving customer to their nearest available driver, such myopic decisions can often lead to suboptimal outcomes.
For instance, after matching a driver to their nearest customer, the next customer request might arrive in that driver's original location, requiring a much longer pickup distance from the remaining available drivers.
Therefore, the platform's key challenge is to design matching policies that minimize the total pickup costs by considering both current matches and potential future request patterns, while maintaining the constraint of immediate matching for each request.

Inspired by the special features of the market dynamics underlying ride-hailing platforms, in this paper, we focus on the online (bipartite) minimum cost matching problem.
In this problem, $m$ nodes from one side of the market are called \emph{servers} and are fully known upfront.
On the other side of the market, which is often referred to as the \emph{request} side, $n$ nodes arrive sequentially.
We model the instantaneity of performing a match by assuming that each arriving request has to be immediately and irrevocably assigned to an available server.
Due to the importance of application to ride-hailing, we assume that all servers and requests are located in the same metric space, and the cost incurred by a match is measured by the distance between two matched nodes.
After realizing a match, both the matched server and request are removed.
The objective of a central market planer is to minimize the total matching cost.

Beyond ride-hailing, online metric matching also finds applications in service platforms, online gaming, and kidney exchanges.
In these settings, participants' features are encoded as high-dimensional vectors, where the distance measures incompatibility---for instance, skill levels in gaming matchmaking, medical compatibility factors in kidney exchanges, or service requirements in platform matching.

We consider both \emph{balanced} and \emph{unbalanced} markets, where balanced markets contain an equal number of servers and requests, and unbalanced markets possess more servers than requests.
Note that balanced markets constitute a special case of unbalanced markets.
We will adopt \emph{competitive ratio} (resp. \emph{regret}) as our performance measure, which is defined as the ratio (resp. difference) between the expected average cost of matching each request incurred by the algorithm and that incurred by the optimal offline solution.

A majority of literature on online metric matching focuses on the \emph{adversarial} model \cite{DBLP:conf/soda/MeyersonNP06,DBLP:conf/compgeom/Raghvendra18,DBLP:journals/talg/PesericoS23}, where the locations of all servers and requests are adversarially chosen, or the \emph{random-order} model \cite{DBLP:conf/approx/Raghvendra16}, which assumes that the arrival order of adversarially chosen requests is randomly permuted.
Both models, albeit being fairly general, suffer from super-constant competitive ratio lower bounds \cite{DBLP:conf/soda/MeyersonNP06,DBLP:conf/approx/Raghvendra16,DBLP:journals/talg/PesericoS23}.
Recent lines of work bypass these barriers by modeling more practical and tractable scenarios in which algorithms are equipped with additional information on future uncertainty.
Specifically, \cite{DBLP:conf/icalp/GuptaGPW19,DBLP:conf/sigecom/BalkanskiFP23} assume that requests are independently drawn from a known distribution, and \cite{DBLP:journals/talg/AntoniadisCEPS23} consider the case where algorithms are provided with predictions of the behaviors of the optimal offline algorithm.
They achieve improved or even nearly optimal results under these weaker models.

\subsection{Main Contribution}

As our main result, we propose a generic framework to designing algorithms for online metric matching that leverages information on future uncertainty provided in the form of predicted request locations.
Our framework operates by enhancing an arbitrary algorithm $\calA$ that does not use additional information.
As a key feature of our framework, the enhancement only invokes $\calA$ in a black-box manner without imposing technical assumptions on $\calA$.

More specifically, given a prediction of request locations, for an arbitrary algorithm $\calA$ that does not use predictions, we can convert $\calA$ into an enhanced algorithm $\calA'$ that bases its matching decisions on both the prediction and $\calA$'s strategy, and the performance of $\calA'$ depends jointly on the prediction accuracy and $\calA$'s performance (\Cref{thm:alg-with-addi-info}).

We then demonstrate the effectiveness of our framework by applying it to the semi-stochastic and prediction settings, following by our additional results for the adversarial setting.

\paragraph{Semi-Stochastic Setting.}
For the semi-stochastic setting where servers are adversarially chosen and requests are independently drawn from a known distribution $\mathbb{D}$, by sampling from $\bbD$ to generate a prediction of request locations, we apply our generic framework to reduce the semi-stochastic setting to the stochastic setting where all servers and requests are independently drawn from $\mathbb{D}$.
Specifically, given an algorithm $\calA$ for the stochastic setting and balanced markets, we can use $\calA$ in a black-box manner to derive an algorithm $\calA'$ for the semi-stochastic setting and unbalanced markets (\Cref{thm:framework-stocas-advers}).
Moreover, if the markets for the semi-stochastic setting are balanced, $\calA'$ inherits the competitive ratio of $\calA$ up to a constant factor.
Next, we describe our concrete results in the semi-stochastic setting obtained by combining our reduction with the recent progress in the stochastic setting~\cite{DBLP:conf/sigecom/Kanoria22,DBLP:conf/sigecom/ChenKKZ23}.

\begin{table}[t]
\centering
\begin{tabular}{@{}lllllll@{}}
\toprule
           & \multicolumn{2}{l}{Balanced Markets} &              &  & \multicolumn{2}{l}{Unbalanced Markets} \\ \cmidrule(lr){2-4} \cmidrule(l){6-7} 
           & Uniform Dis.   & Smooth Dis.         & General Dis. &  & Smooth Dis.           & General Dis.   \\ \midrule
$d=1$      & $O(1)$         & $O(1)$              & $O(1)$       &  & \textcolor{red}{$O(\log n)$}           & \textcolor{red}{$O(\log n)$}    \\
$d=2$      & \textcolor{red}{$O(1)$} ($\omega(1)$)         & $\omega(1)$  & $\omega(1)$     &  & \textcolor{red}{$O(\sqrt{\log n})$}    & \textcolor{red}{$O(\log n)$}       \\
$d \geq 3$ & \textcolor{red}{$O(1)$} ($\omega(1)$)         & \textcolor{red}{$O(1)$} ($\omega(1)$)              & $\omega(1)$     &  & \textcolor{red}{$O(1)$}                & \textcolor{red}{$O(\log n)$}       \\ \bottomrule
\end{tabular}
\caption{Competitive ratios for the Euclidean metric $[0, 1]^d$ in the semi-stochastic setting.
All previous state-of-the-art results are given in black fonts, where we use $\omega(1)$ to denote $O((\log \log \log n)^2)$, and all new results achieved in this paper are marked in red.
In particular, for balanced markets with general distributions, the competitive ratios of $O(1)$ for $d = 1$ and of $O((\log \log \log n)^2)$ for $d \geq 2$ are achieved by \cite{DBLP:conf/icalp/GuptaGPW19}, which imply the bounds on their left.}
\label{tab:results-cr}
\end{table}

We first introduce the competitive ratio results for the Euclidean metric $[0, 1]^d$, which are summarized in \Cref{tab:results-cr}.
Prior to our work, the best competitive ratio achieved for balanced markets with general metrics and distributions is $O((\log \log \log n)^2)$, and the same algorithm is $O(1)$-competitive for tree metrics~\cite{DBLP:conf/icalp/GuptaGPW19}.
We start by showing that combining our reduction with the algorithm of \cite{DBLP:conf/sigecom/Kanoria22} leads to an $O(1)$-competitive algorithm for all $d \geq 1$ in balanced markets and the semi-stochastic setting with the uniform distribution (\Cref{thm:bal-uni}).

Next, we describe our results in balanced markets and the semi-stochastic setting with smooth distributions, where a distribution is smooth if it admits a density function upper bounded by a constant.
For the algorithm $\calA$ required by the reduction, we adopt the \emph{Simulate-Optimize-Assign-Repeat (SOAR)} algorithm of \cite{DBLP:conf/sigecom/ChenKKZ23}, which results in an $O(1)$-competitive algorithm for $d \geq 3$ (\Cref{thm:smooth-adv-cr}).

Finally, we present our results for the semi-stochastic setting and unbalanced markets, for which we assume that servers outnumber requests by at most a constant factor.
For smooth distributions, by applying the same approach as for balanced markets, the resulting algorithm admits a competitive ratio of $O(\sqrt{n})$ for $d = 1$, $O(\sqrt{\log n})$ for $d = 2$, and $O(1)$ for $d \geq 3$ (\Cref{thm:ratio-ahg-excess}).
To improve the undesirable bound for $d = 1$, we show that the $O(\log n)$ competitive ratio of the algorithm proposed by \cite{DBLP:conf/icalp/GuptaL12}, which holds in balanced markets when servers and requests are all adversarially located in a doubling metric\footnote{A metric space is \emph{doubling} if there exists a constant $M > 0$ such that for all $r > 0$, every ball of radius $r$ can be covered by at most $M$ balls of radius $r / 2$. In particular, the Euclidean metric $[0, 1]^d$ is doubling for constant $d$.}, is preserved in unbalanced markets (\Cref{cor:adv-unbal-line}).
These constitute the first known results for the semi-stochastic setting and unbalanced markets.

\begin{table}[t]
\centering
\begin{tabular}{@{}llllll@{}}
\toprule
           & \multicolumn{2}{l}{Stoch. and Balanced}                      &  & \multicolumn{2}{l}{Semi-Stoch. and Unbalanced}                           \\ \midrule
           & Lower Bound                & Upper Bound                  &  & Lower Bound                       & Upper Bound                  \\ \midrule
$d=1$      & $\Omega(n^{-\frac{1}{2}})$ & $O(n^{-\frac{1}{2}})$        &  & $\Omega(n^{-\frac{1}{2}})$        & \textcolor{red}{$O(n^{-\frac{1}{2}})$}        \\
$d=2$      & $\Omega(n^{-\frac{1}{2}})$ & $O(n^{-\frac{1}{2}} \sqrt{\log n})$ &  & $\Omega(n^{-\frac{1}{2}})$ & \textcolor{red}{$O(n^{-\frac{1}{2}} \sqrt{\log n})$} \\
$d \geq 3$ & $\Omega(n^{-\frac{1}{d}})$ & $O(n^{-\frac{1}{d}})$        &  & $\Omega(n^{-\frac{1}{d}})$        & \textcolor{red}{$O(n^{-\frac{1}{d}})$}        \\ \bottomrule
\end{tabular}
\caption{Regret for the Euclidean metric $[0, 1]^d$ with general distributions.
All previous results are given in black fonts, and all new results achieved in this paper are marked in red.
In particular, the results for the stochastic setting and balanced markets are achieved by \cite{DBLP:conf/sigecom/ChenKKZ23}, which imply the lower bounds for the semi-stochastic setting and unbalanced markets.}
\label{tab:results-regret}
\end{table}

We further describe our results on regret analysis for the Euclidean metric $[0, 1]^d$, which are summarized in~\Cref{tab:results-regret}.
Recall that for balanced markets and the stochastic setting with general distributions, \cite{DBLP:conf/sigecom/ChenKKZ23} show that the SOAR algorithm admits regret of $O(n^{-\frac{1}{2}})$ for $d = 1$, $O(n^{-\frac{1}{2}} \sqrt{\log n})$ for $d = 2$, and $O(n^{-\frac{1}{d}})$ for $d \geq 3$, which is optimal for $d \neq 2$ and optimal up to a logarithmic factor for $d = 2$.
By combining our reduction with the SOAR algorithm, we show that the same regret is attainable even for unbalanced markets and the semi-stochastic setting with general distributions (\Cref{thm:any-dis-regret-ub}).
Since the case with balanced markets and the stochastic setting constitutes a special case of unbalanced markets and the semi-stochastic setting, the lower bounds by \cite{DBLP:conf/sigecom/ChenKKZ23} imply our regret guarantees to be nearly optimal.

\paragraph{Prediction Setting.}
In the prediction setting, the algorithm receives a set of $n$ points as a prediction in advance, which predicts the unknown request locations.
We consider balanced markets and the adversarial setting, where servers and requests are all adversarially chosen.
In particular, the prediction is not guaranteed to be accurate, and its error is oblivious to the algorithm.
The algorithm is expected to be (1) \emph{consistent}: its performance is nearly optimal with perfectly accurate predictions, and (2) \emph{robust}: with inaccurate predictions, its performance is comparable to the best online algorithm that does not use predictions.
In our case, the prediction error $\eta$ is defined as the minimum cost of all perfect matchings between the prediction and the realized request locations.

Given an $\alpha$-competitive algorithm $\calA$ that does not use predictions, we apply our generic framework to derive an algorithm $\calA'$ that leverages the prediction and achieves a competitive ratio of $O(\min\{\alpha, 1+ \frac{\eta (\alpha + 1) }{\opt}\})$, where $\opt$ denotes the optimal offline cost (\Cref{thm:comp-alg}).
To see the consistency of $\calA'$, its competitive ratio converges to a constant when the prediction error $\eta$ approaches $0$.
Moreover, under a large prediction error, we recover the competitive ratio $\alpha$ provided by $\calA$ up to a constant factor, fulfilling the robustness requirement.
Besides, the performance of $\calA'$ degrades smoothly as the prediction error increases, which is usually referred to as the \emph{smoothness} property.
Finally, we apply the above result to several specific metrics (\Cref{cor:pred-metric}).

To highlight the practical relevance of our prediction model, in reality, ride-hailing platforms may not predict exact future request locations but rather use historical data to estimate demand distributions across regions. Our prediction model with point locations serves as a useful abstraction that captures various practical scenarios: (1) predictions could be sampled from estimated distributions, providing robustness when these estimates are imperfect—unlike the semi-stochastic setting which assumes perfect distributional knowledge; (2) predictions could represent anticipated requests based on scheduled events or historical patterns; (3) the prediction error $\eta$ naturally quantifies uncertainty in any forecasting method.

Our framework's strength lies in its generality—it provides performance guarantees regardless of how predictions are generated, degrading gracefully from near-optimal performance with accurate predictions to worst-case guarantees with poor predictions. This robustness is particularly valuable when distributional assumptions may not hold perfectly in practice.

\paragraph{Adversarial Setting.}

We turn to our results for the adversarial setting without predictions.
As alluded to previously, we generalize the $O(\log n)$ competitive ratio of the algorithm by \cite{DBLP:conf/icalp/GuptaL12} for doubling metrics from balanced markets to unbalanced markets.
This is achieved based on the more general result that every greedy algorithm with a tie-breaking rule satisfying certain properties admits the same competitive ratio in balanced and unbalanced markets (\Cref{thm:greedy-bal-to-unbal}).
As a by-product, we also show that the $O(\log^3 n)$-competitive algorithm by \cite{DBLP:conf/soda/MeyersonNP06} for balanced markets with general metrics retains its competitive ratio in unbalanced markets (\Cref{cor:adv-arbi-metric}).
We remark here that all prior results in the adversarial setting are established in balanced markets, and it is not immediate whether they can be carried over to unbalanced markets.

\subsection{Related Work}
\label{sec:related-work}

\paragraph{Online Metric Matching in the Adversarial Setting.}
A majority of literature for online metric matching focuses on the adversarial setting and balanced markets.
For general metrics, \cite{DBLP:journals/jal/KalyanasundaramP93}, \cite{DBLP:journals/tcs/KhullerMV94}, and \cite{DBLP:conf/approx/Raghvendra16} respectively give $(2n - 1)$-competitive deterministic algorithms, which is optimal.
When randomization is allowed, an $O(\log^3 n)$-competitive algorithm is proposed by \cite{DBLP:conf/soda/MeyersonNP06}, which is subsequently improved to $O(\log^2 n)$~\cite{DBLP:journals/algorithmica/BansalBGN14}.
In the random-order model, \cite{DBLP:conf/approx/Raghvendra16} proposes an $O(\log n)$-competitive deterministic algorithm, which is optimal even for randomized algorithms, and \cite{DBLP:journals/orl/GairingK19} show that the competitive ratio of the greedy algorithm lies between $\Omega(n^{0.22})$ and $O(n)$.

For the line metric, \cite{DBLP:journals/jal/KalyanasundaramP93} show that the competitive ratio of the greedy algorithm is at least $2^n - 1$, and a deterministic algorithm with a sublinear competitive ratio is known \cite{DBLP:journals/algorithmica/AntoniadisBNPS19}.
Furthermore, \cite{DBLP:conf/focs/NayyarR17} present an $O(\log^2 n)$-competitive deterministic algorithm, which is later shown to be $O(\log n)$-competitive \cite{DBLP:conf/compgeom/Raghvendra18}.
\cite{DBLP:conf/icalp/GuptaL12} propose an $O(\log n)$-competitive randomized algorithm for the more general doubling metrics.
Regarding lower bounds, \cite{DBLP:journals/tcs/FuchsHK05} show that no deterministic algorithm can achieve a competitive ratio strictly better than $9.001$; it is then improved to $\Omega(\sqrt{\log n})$, which holds even for randomized algorithms \cite{DBLP:journals/talg/PesericoS23}.

\paragraph{Online Metric Matching in the Stochastic Setting.}
We start by reviewing the literature for balanced markets.
For the Euclidean metric
$[0, 1]^d$, \cite{DBLP:conf/sigecom/Kanoria22} presents an $O(1)$-competitive algorithm for the uniform distribution, and \cite{DBLP:conf/sigecom/ChenKKZ23} give an algorithm achieving nearly optimal regret for general distributions.
\cite{DBLP:conf/sigecom/BalkanskiFP23} show that the greedy algorithm is $O(1)$-competitive for $[0, 1]$ with the uniform distribution.

Motivated by the potential value of excess servers, \cite{DBLP:conf/sigecom/AkbarpourALS22} initiate the study of unbalanced markets, for which they show that the greedy algorithm is $O(\log^3 n)$-competitive for $[0, 1]$ with the uniform distribution.
This is subsequently improved to $O(1)$~\cite{DBLP:conf/sigecom/BalkanskiFP23}.
\cite{DBLP:conf/sigecom/Kanoria22} gives a $O(1)$-competitive algorithm for unbalanced markets and the Euclidean metric $[0, 1]^d$ with the uniform distribution.

\paragraph{Online Metric Matching in the Semi-Stochastic Setting.}
In contrast to the fruitful development in the stochastic setting, much less is known for the semi-stochastic setting, for which all existing results are for balanced markets.
For general distributions, \cite{DBLP:conf/icalp/GuptaGPW19} give an $O((\log \log \log n)^2)$-competitive algorithm for general metrics, which is $O(1)$-competitive for tree metrics.
Moreover, for $[0, 1]$ with the uniform distribution, \cite{DBLP:conf/sigecom/BalkanskiFP23} show that the greedy algorithm is $\Theta(\log n)$-competitive.

\paragraph{Online Algorithm with Predictions.}
In recent years, due to the unparalleled impact of machine learning, the paradigm of designing online algorithms with access to externally provided predictions is of particular interest to overcome both information-theoretic and computational barriers.
This paradigm suitably models the practical scenarios where future uncertainty is predictable according to prior knowledge, where predictions can be generated by machine learning models, human experts, historical data, heuristics, or even crowdsourcing.
In particular, it is commonly assumed that the provided predictions are untrusted with their error oblivious to the algorithm, and the algorithm is desired to be consistent and robust.
This paradigm has been incorporated into the study of various classical problems including online metric matching \cite{DBLP:journals/talg/AntoniadisCEPS23}, metrical task systems \cite{DBLP:conf/icml/0001C00S23,DBLP:journals/talg/AntoniadisCEPS23}, online facility location \cite{DBLP:conf/nips/AlmanzaCLPR21,DBLP:conf/iclr/JiangLLTZ22}, online bipartite matching \cite{DBLP:conf/nips/JinM22,DBLP:journals/disopt/AntoniadisGKK23}, and caching \cite{DBLP:journals/jacm/LykourisV21,DBLP:conf/nips/0001PS022,DBLP:conf/icml/Anand0KP22}, among others.

The paper most relevant to our work is \cite{DBLP:journals/talg/AntoniadisCEPS23}, which also studies online metric matching with predictions.
As the primary difference, we assume that the prediction is a set of request locations and is given in advance, whereas they assume that each arriving request is accompanied by a prediction of the behaviors of the optimal offline algorithm until the end of the current period, i.e., all servers matched by the optimal offline algorithm to the requests arrived so far.\footnote{They allow the predictions associated with different requests to be inconsistent, i.e., each prediction may not be a superset of previous predictions.}
Moreover, their prediction error is defined against an optimal offline algorithm as the summation of prediction errors of all periods, where the prediction error of each period is defined as the minimum cost of all perfect matchings between the predicted matched servers and the actual matched servers.
Hence, their results are incomparable with ours.
\section{Model Setup}

In our model, $n$ requests $R = (r_1, \ldots, r_n)$ and $m$ servers $S = \{s_1, \ldots, s_m\}$ are located in a metric space $(X, \delta)$ with $n \leq m$.
The servers are all available at time $t = 0$ with their locations known to the algorithm.
At each time step $t \in [n]$, request $r_t$ arrives with its location revealed to the algorithm and needs to be matched immediately and irrevocably to a server $s$ that has not been matched, incurring a cost $\delta(s, r_t)$.

Depending on the context, we will assume the locations of servers and requests are either adversarially chosen or independently drawn from a known distribution $\mathbb{D}$.
In the \emph{adversarial} setting, both $S$ and $R$ are adversarially chosen.
In the \emph{semi-stochastic} setting, $S$ is adversarial, while the requests in $R$ are independently drawn from $\bbD$.
In the \emph{stochastic} setting, all servers in $S$ and requests in $R$ are independently drawn from $\bbD$.

For all algorithm $\calA$, set of servers $S$ with $|S| \geq n$, and tuple of requests $R$ with $|R| = n$, denote $\cost(\calA; S, R)$ as the expected cost of $\calA$ under $S$ and $R$; denote $\costD(\calA; S, n)$ as the expected cost of $\calA$ under $S$ and $n$ requests independently drawn from $\mathbb{D}$; denote $\costD(\calA; n)$ as the expected cost of $\calA$ under $n$ servers and $n$ requests, all independently drawn from $\mathbb{D}$.
Given a matching $M$ between two sets $S$ and $R$, for each matched element $s \in S$ (resp. $r \in R$), we use $M(s)$ (resp. $M(r)$) to denote the element in $R$ (resp. $S$) that $s$ (resp. $r$) is matched to.

We define the smoothness property for distributions as follows.

\begin{definition}[Smoothness]
    We say that a distribution $\mathbb{D}$ over $X$, which supports a uniform measure $\mathbb{U}$, is \emph{$\beta$-smooth} for $\beta \geq 1$ if the corresponding measure $\mu_{\bbD}$ satisfies $\mu_{\mathbb{D}}(X') \leq \beta \cdot \mathbb{U}(X')$ for every measurable subset $X' \subseteq X$.
\end{definition}

\subsection{Benchmark}

We adopt the optimal offline cost, defined as the optimal matching cost without any future uncertainty, as our benchmark.
For all set of servers $S$ and tuple of requests $R$ such that $|S| \geq |R|$, denote $\opt(S, R)$ as the minimum cost of all matchings between $S$ and $R$ that match all requests.
Given a set of servers $S$ with $|S| \geq n$ and a distribution $\mathbb{D}$ over $X$, define $\optD(S, n) := \Expect_{R \sim \mathbb{D}^n} \left[ \opt(S, R) \right]$ as the optimal offline cost on $S$ with $n$ requests in $R$ independently drawn from $\mathbb{D}$.
Furthermore, define $\optD(n) := \Expect_{S \sim \bbD^n, R \sim \mathbb{D}^n} \left[ \opt(S, R) \right]$.

We use \emph{competitive ratio} and \emph{regret} to measure the performance of online algorithms.
In the adversarial setting, we say that an algorithm $\calA$ is $\alpha$-competitive if for all $S$ and $R$, $\cost(\calA; S, R) \leq \alpha \cdot \opt(S, R)$.
In the semi-stochastic setting, we say that an algorithm $\calA$ is $\alpha$-competitive if for all $\bbD$, $S$, and $n$, $\costD(\calA; S, n) \leq \alpha \cdot \optD(S, n)$; also, define the regret of $\calA$ as
\begin{align*}
    \regret(\calA; S, n)
    := \frac{1}{n} \left( \costD(\calA; S, n) - \optD(S, n) \right)  \,.
\end{align*}
In the stochastic setting, we say that an algorithm $\calA$ is $\alpha$-competitive if for all $\bbD$ and $n$, $\costD(\calA; n) \leq \alpha \cdot \optD(n)$.

\subsection{Benchmark for Euclidean Metric}
\label{sec:optoff-euc}

A considerable portion of this paper focuses on the Euclidean metric.
In this subsection, we assume the metric space to be $[0, 1]^d$ with the Euclidean distance, and we bound the optimal offline costs for smooth distributions and general distributions.

For an arbitrary distribution $\mathbb{D}$, the following upper bounds for $\optD(n)$ are known from the optimal transport literature.

\begin{lemma}[\cite{talagrand1992ajtai,dobric1995asymptotics,bobkov2019one}]\label{lem:optoff-any-ub}
    For every distribution $\mathbb{D}$ over $[0, 1]^d$,
    \begin{align*}
        \optD(n)
        \leq \begin{cases}
            O(\sqrt{n}) & d = 1\\
            O(\sqrt{n \log n}) & d = 2\\
            O(n^{1 - \frac{1}{d}}) & d \geq 3
        \end{cases} \,.
    \end{align*}
\end{lemma}

Next, we establish lower bounds for $\optD(n)$ with smooth distribution $\bbD$.
We defer the proof of \Cref{lem:optoff-smooth-lb} to \Cref{sec:proof-lem-optoff-smooth-lb}.

\begin{lemma}\label{lem:optoff-smooth-lb}
    Assume $\mathbb{D}$ to be a $\beta$-smooth distribution over $[0, 1]^d$ for $\beta \geq 1$.
    Then,
    \begin{align*}
        \optD(n) \geq
        \begin{cases}
            \Omega(\beta^{-1} \sqrt{n}) & d = 1\\
            \Omega(\beta^{-\frac{1}{d}} n^{1 - \frac{1}{d}}) & d \geq 2 
        \end{cases}\,.
    \end{align*}
\end{lemma}

The lower bound in \Cref{lem:optoff-smooth-lb} for $d = 1$ follows from the quantile mapping and the well-known $\Omega(\sqrt{n})$ lower bound for the optimal offline cost with the uniform distribution, and the lower bound for $d \geq 2$ relies on nearest-neighbor-distance.
For $d = 2$, we fail to generalize the lower bound of $\Omega(\sqrt{n \log n})$ in the Ajtai-Koml\'os-Tusn\'ady Matching Theorem, whose proof crucially exploits the structures of the uniform distribution.
We remark that the optimal matching for $d = 2$ exhibits the most intricate structure, and we refer to \cite{DBLP:conf/sigecom/Kanoria22,talagrand2022upper} for more thorough discussion.
Hence, we leave the task of closing the gap for $d = 2$ to future work.

Finally, we present lower bounds for unbalanced markets with smooth distributions, whose proof also employs nearest-neighbor-distance.
We defer the proof of \Cref{lem:lb-optoff-excess} to \Cref{sec:proof-lem-lb-optoff-excess}.

\begin{lemma}\label{lem:lb-optoff-excess}
    Assume $\mathbb{D}$ to be a $\beta$-smooth distribution over $[0, 1]^d$ for $\beta \geq 1$.
    Then, for all $d \geq 1$ and $S \subseteq [0, 1]^d$ with $|S| = \kappa n$ and $\kappa \geq 1$, $\optD(S, n) \geq \Omega((\beta\kappa)^{-\frac{1}{d}} n^{1 - \frac{1}{d}})$.
\end{lemma}

\subsection{The Simulate-Optimize-Assign-Repeat Algorithm}

The \emph{Simulate-Optimize-Assign-Repeat (SOAR)} algorithm is proposed by \cite{DBLP:conf/sigecom/ChenKKZ23} in the stochastic setting and balanced markets.
It will serve as an important building block of our algorithms in the semi-stochastic setting, and we state its guarantee as follows.

\begin{lemma}[\cite{DBLP:conf/sigecom/ChenKKZ23}]\label{lem:matching-cost-soar}
    Suppose that $m = n$ and $\mathbb{D}$ is an arbitrary distribution over $X$, where the cost function $\delta$ may \emph{not} satisfy the triangle inequality.
    Then, when servers and requests are independently drawn from $\mathbb{D}$, there exists an algorithm with expected cost $\sum_{t=1}^n \frac{1}{t} \cdot \optD(t)$.
\end{lemma}
\section{Generic Framework for Leveraging Predictions of Request Locations}
\label{sec:reduction}

In this section, we present a generic framework for online metric matching to design algorithms that leverage the provided prediction in the form of a set of request locations, which operates by enhancing an algorithm that does not use predictions.
Then, we apply this framework to the semi-stochastic setting.

\begin{theorem}\label{thm:alg-with-addi-info}
    Given an algorithm $\calA$ for online metric matching, there exists an algorithm $\calA'$, which is provided with a set of points $I$ before the arrival of any request, such that for all set of servers $S$ and tuple of requests $R$ with $|S| \geq |R| = |I|$, $\cost(\calA'; S, R) \leq \opt(S, I) + \cost(\calA; I, R)$.
    Moreover, if $\calA$ is deterministic, then $\calA'$ is also deterministic.
\end{theorem}

\begin{proof}
    We construct the algorithm $\calA'$ as follows.
    Fix $S, I, R$ as required.
    Let $M^{S, I}$ be the matching between $S$ and $I$ with the minimum cost that matches all points in $I$.
    Then, we run $\calA$ with the set of initial servers being $I$.
    For each new request $r$, if $\calA$ matches $r$ to $i \in I$, then $\calA'$ matches $r$ to $M^{S, I}(i) \in S$.

    Since $\calA'$ does not use any randomness except for the invocation of $\calA$, $\calA'$ is deterministic provided that $\calA$ is deterministic.
    By the triangle inequality, the cost incurred by $\calA'$ is upper bounded by
    \begin{align*}
        \cost(\calA'; S, R)
        \leq \sum_{i \in I} \delta(M^{S, I}(i), i) + \cost(\calA; I, R)
        =  \opt(S, I) + \cost(\calA; I, R)\,,
    \end{align*}
    concluding the proof.
\end{proof}

\subsection{Semi-Stochastic Setting}
\label{sec:stoch-setting-pre}

In this subsection, we apply \Cref{thm:alg-with-addi-info} to the semi-stochastic setting, which enables a reduction to the stochastic setting.
Specifically, we establish in the following theorem that an algorithm for the stochastic setting can be transformed into an algorithm for the semi-stochastic setting in a black-box manner.
In particular, for balanced markets, this transformation preserves the competitive ratio of the algorithm up to a constant factor.

\begin{theorem}\label{thm:framework-stocas-advers}
    For every distribution $\bbD$, given an algorithm $\calA$ for online metric matching, there exists an algorithm $\calA'$ such that $\costD(\calA'; S, n) \leq \optD(S, n) + \costD(\calA; n)$ for every set of servers $S$ with $|S| \geq n$.
    Moreover, for balanced markets, if $\calA$ is $\alpha$-competitive in the stochastic setting, then $\calA'$ is $(2 \alpha + 1)$-competitive in the semi-stochastic setting.
\end{theorem}

The rest of this subsection is devoted to proving \Cref{thm:framework-stocas-advers}.
Informally, we will construct the algorithm $\calA'$ using \Cref{thm:alg-with-addi-info}, where the points in the set $I$ are independently sampled from $\mathbb{D}$.
Then, we prove the guarantees of $\calA'$ by applying the triangle inequality.

We start with proving the triangle inequality for $\opt(\cdot, \cdot)$.

\begin{lemma}\label{lem:tri-opt}
    For all sets of points $S, I, R$ with $|S| = |I| = |R|$, $\opt(S, R) \leq \opt(S, I) + \opt(I, R)$.
\end{lemma}

\begin{proof}
    Let $M_1$ be the min-cost perfect matching between $S$ and $I$, and let $M_2$ be the min-cost perfect matching between $I$ and $R$.
    Note that
    \begin{align*}
        \opt(S, I) + \opt(I, R)
        = \sum_{i \in I} \delta(M_1(i), i) + \sum_{i \in I} \delta(i, M_2(i))
        \geq \sum_{i \in I} \delta(M_1(i), M_2(i))
        \geq \opt(S, R)  \,,
    \end{align*}
    where the first inequality holds by the triangle inequality.
\end{proof}

Next, for balanced markets, we lower bound the optimal offline cost in the semi-stochastic setting in terms of that in the stochastic setting.

\begin{lemma}\label{lem:lb-optoff-by-optoff}
    For all distribution $\mathbb{D}$ and set of servers $S$ with $|S| = n$, $\optD(n) \leq 2 \cdot \optD(S, n)$.
\end{lemma}

\begin{proof}
    Given a set of servers $S$ with $|S| = n$, by \Cref{lem:tri-opt}, $\opt(I, R) \leq \opt(I, S) + \opt(S, R)$ for all sets of points $I, R$ with $|I| = |R| = n$.
    Taking expectations on both sides over $I \sim \mathbb{D}^n$ and $R \sim \mathbb{D}^n$, we obtain
    \begin{align*}
        \optD(n)
        \leq \Expect_{I \sim \bbD^n, R \sim \mathbb{D}^n} \left[ \opt(I, S) + \opt(S, R) \right]
        = 2 \cdot \optD(S, n) \,,
    \end{align*}
    which concludes the proof.
\end{proof}

Now, we are ready to prove \Cref{thm:framework-stocas-advers}.

\begin{proof}[Proof of \Cref{thm:framework-stocas-advers}.]
    We construct the algorithm $\calA'$ as follows.
    Given as input a set of servers $S$ and a tuple of requests $R$ with $|S| \geq |R| = n$, we sample a set of points $I \sim \mathbb{D}^n$ before the arrival of any request, and let $\calA'$ be the algorithm as per \Cref{thm:alg-with-addi-info}, which satisfies
    \begin{align*}
        \cost(\calA'; S, R)
        \leq \opt(S, I) + \cost(\calA; I, R)\,.
    \end{align*}
    Since $I$ and $R$ are independent, the expected cost of $\calA'$ is upper bounded by
    \begin{align*}
        \costD(\calA'; S, n)
        \leq \Expect_{I \sim \bbD^n, R \sim \mathbb{D}^n} \left[ \opt(S, I) + \cost(\calA; I, R) \right]
        =  \optD(S, n) + \costD(\calA; n)\,,
    \end{align*}
    which concludes the first part in \Cref{thm:framework-stocas-advers}.
    Furthermore, for balanced markets, if $\calA$ is $\alpha$-competitive in the stochastic setting, then
    \begin{align*}
        \costD(\calA'; S, n)
        \leq \optD(S, n) + \alpha \cdot \optD(n)
        \leq (2\alpha + 1) \cdot \optD(S, n) \,,
    \end{align*}
    where the second inequality holds by \Cref{lem:lb-optoff-by-optoff}, concluding the second part in \Cref{thm:framework-stocas-advers}.
\end{proof}
\section{Semi-Stochastic Setting: Euclidean Metric with Smooth Distributions}
\label{sec:euc-smooth}

In this section, we consider the semi-stochastic setting.
Moreover, we assume the metric space to be $[0, 1]^d$ with the Euclidean distance and $\mathbb{D}$ to be $\beta$-smooth\footnote{The smoothness of distributions is a mild technical assumption allowing us to lower bound $\optD(n)$. To see the necessity, consider $\mathbb{D}$ a non-smooth distribution supported on a singleton, in which case $\optD(n) = 0$.} for some constant $\beta \geq 1$, and we apply \Cref{thm:framework-stocas-advers} to derive algorithms for both balanced and unbalanced markets.

In \Cref{sec:bal-euc-uni}, for constant $d \geq 1$, we give an $O(1)$-competitive algorithm for balanced markets wth the uniform distribution.

Next, we consider balanced markets with smooth distributions in \Cref{sec:bal-euc-smooth}.
Recall that for general distributions, \cite{DBLP:conf/icalp/GuptaGPW19} give an $O(1)$-competitive algorithm for balanced markets and $d = 1$, which is $O((\log \log \log n)^2)$-competitive for general metrics.
However, for $d = 2$ and smooth distributions, we are only able to estimate the optimal offline cost up to a $\Theta(\sqrt{\log n})$ factor, which prevents our analysis from achieving an $o(\sqrt{\log n})$ competitive ratio. 
Hence, we only focus on the case with $d \geq 3$, where we present an $O(1)$-competitive algorithm.

Finally, in \Cref{sec:unbal-smooth}, we extend our results for balanced markets to unbalanced markets that possess at most a constant factor of excess servers, i.e., $n \leq m \leq O(n)$.
We give an algorithm with a competitive ratio of $O(\sqrt{n})$ for $d = 1$, $O(\sqrt{\log n})$ for $d = 2$, and $O(1)$ for $d \geq 3$.

\subsection{Prelude: Balanced Markets with Uniform Distribution}
\label{sec:bal-euc-uni}

In this subsection, we consider balanced markets and assume $\mathbb{D}$ to be the uniform distribution.
In the following theorem, we combine \Cref{thm:framework-stocas-advers} with the algorithm of \cite{DBLP:conf/sigecom/Kanoria22} to give an $O(1)$-competitive algorithm for constant $d \geq 1$.

\begin{theorem}\label{thm:bal-uni}
    Suppose that $m = n$ and $\bbD$ is the uniform distribution over $[0, 1]^d$.
    Then, for $d \geq 1$, there exists an $O(\sqrt{d})$-competitive algorithm in the semi-stochastic setting.
\end{theorem}

\begin{proof}
    We first recall the guarantees of the algorithm proposed by \cite{DBLP:conf/sigecom/Kanoria22}.

    \begin{theorem}[\cite{DBLP:conf/sigecom/Kanoria22}]\label{thm:alg-hg}
        Suppose that $m = n$ and $\bbD$ is the uniform distribution over $[0, 1]^d$.
        Then, for $d \geq 1$, there exists an $O(\sqrt{d})$-competitive algorithm in the stochastic setting.
    \end{theorem}
    
    Combining Theorems~\ref{thm:framework-stocas-advers} and~\ref{thm:alg-hg} concludes the proof.
\end{proof}

\subsection{Balanced Markets}
\label{sec:bal-euc-smooth}

In this subsection, we assume $d \geq 3$ and consider balanced markets with smooth distributions.
We prove that the SOAR algorithm is $O(1)$-competitive for $d \geq 3$ in the stochastic setting.
Combining with \Cref{thm:framework-stocas-advers} leads to an $O(1)$-competitive algorithm in the semi-stochastic setting with $d \geq 3$.

We first present our main theorem in this subsection.

\begin{theorem}\label{thm:smooth-adv-cr}
    Suppose that $m = n$ and $\mathbb{D}$ is a $\beta$-smooth distribution over $[0, 1]^d$ for $\beta \geq 1$.
    Then, for $d \geq 3$, there exists an $O(\beta^{\frac{1}{d}})$-competitive algorithm in the semi-stochastic setting.
\end{theorem}

The rest of this subsection is devoted to proving \Cref{thm:smooth-adv-cr}.
By \Cref{thm:framework-stocas-advers}, it suffices to show that the SOAR algorithm is $O(\beta^{\frac{1}{d}})$-competitive for $d \geq 3$ in the stochastic setting.
We compute the expected cost of the SOAR algorithm in the following lemma.

\begin{lemma}\label{lem:ratio-soar}
    Suppose that $m = n$ and $\mathbb{D}$ is an arbitrary distribution over $[0, 1]^d$.
    Then, there exists an algorithm $\calA$ such that
    \begin{align*}
        \costD(\calA; n) \leq
        \begin{cases}
            O(\sqrt{n}) & d = 1\\
            O(\sqrt{n \log n}) & d = 2\\
            O(n^{1 - \frac{1}{d}}) & d \geq 3
        \end{cases} \,.
    \end{align*}
\end{lemma}

\begin{proof}
    Let $\calA$ be the algorithm as per \Cref{lem:matching-cost-soar}, and we estimate the cost of $\calA$ for $d = 1$, $d = 2$, and $d \geq 3$ separately.
    Recall that an upper bound for $\optD(n)$ is given by \Cref{lem:optoff-any-ub}.
    For $d = 1$,
    \begin{align*}
        \cost_{\mathbb{D}}(\calA; n)
        = \sum_{t=1}^n \frac{1}{t} \cdot \opt_{\mathbb{D}}(t)
        = O \left( \sum_{t=1}^n \frac{1}{\sqrt{t}} \right)
        = O \left( \int_1^n \frac{1}{\sqrt{x}} \mathrm{d}x \right)
        = O\left( \sqrt{n} \right) \,.
    \end{align*}
    For $d = 2$,
    \begin{align*}
        \cost_{\mathbb{D}}(\calA; n)
        = \sum_{t=1}^n \frac{1}{t} \cdot \opt_{\mathbb{D}}(t)
        = O \left( \sum_{t=1}^n \sqrt{\frac{\log t}{t}} \right)
        = O \left( \sqrt{\log n} \cdot \int_1^n \frac{1}{\sqrt{x}} \mathrm{d}x \right)
        = O \left(\sqrt{n \log n}\right) \,.
    \end{align*}
    For $d \geq 3$,
    \begin{align*}
        \cost_{\mathbb{D}}(\calA; n)
        = \sum_{t=1}^n \frac{1}{t} \cdot \opt_{\mathbb{D}}(t)
        = O \left( \sum_{t=1}^n \frac{1}{t^{1 / d}} \right)
        = O \left( \int_1^n \frac{1}{x^{1 / d}} \mathrm{d}x \right)
        = O\left( n^{1 - \frac{1}{d}} \right) \,.
    \end{align*}
    Combining all the above concludes the proof.
\end{proof}

Finally, combining \Cref{lem:ratio-soar} with the lower bounds for $\optD(n)$ given by \Cref{lem:optoff-smooth-lb} concludes the proof of \Cref{thm:smooth-adv-cr}.

\subsection{Unbalanced Markets}
\label{sec:unbal-smooth}

We extend our results for balanced markets to unbalanced markets in the following theorem.

\begin{theorem}\label{thm:ratio-ahg-excess}
    Suppose that $m = \kappa n$ and $\mathbb{D}$ is a $\beta$-smooth distribution over $[0, 1]^d$ for $\kappa \geq 1$ and $\beta \geq 1$.
    Then, in the semi-stochastic setting, there exists an algorithm with a competitive ratio of $O(\beta\kappa\sqrt{n})$ for $d = 1$, $O(\sqrt{\beta\kappa\log n})$ for $d = 2$, and $O((\beta\kappa)^{\frac{1}{d}})$ for $d \geq 3$.
\end{theorem}

\begin{proof}
    Let $\calA$ be the algorithm as per \Cref{lem:ratio-soar}.
    By \Cref{thm:framework-stocas-advers}, there exists an algorithm $\calA'$ such that $\costD(\calA'; S, n) \leq \optD(S, n) + \costD(\calA; n)$ for every set of servers $S$ with $|S| \geq n$.
    Combining the lower bounds for $\optD(S, n)$ given by \Cref{lem:lb-optoff-excess} concludes the proof.
\end{proof}
\section{Semi-Stochastic Setting: Euclidean Metric with General Distributions}
\label{sec:euc-arbi}

In this section, we consider the semi-stochastic setting and assume the metric space to be $[0, 1]^d$ with the Euclidean distance.
For general distributions, we give an algorithm for unbalanced markets with an arbitrary number of excess servers, whose regret nearly matches the lower bounds established by \cite{DBLP:conf/sigecom/ChenKKZ23} for balanced markets and the stochastic setting.

\begin{theorem}\label{thm:any-dis-regret-ub}
    Suppose that $\mathbb{D}$ is an arbitrary distribution over $[0, 1]^d$.
    Then, there exists an algorithm $\calA'$ with regret
    \begin{align*}
        \regret(\calA'; S, n)
        \leq \begin{cases}
            O(n^{-\frac{1}{2}}) & d = 1\\
            O(n^{-\frac{1}{2}} \sqrt{\log n}) & d = 2\\
            O(n^{- \frac{1}{d}}) & d \geq 3
        \end{cases}
    \end{align*}
    for every set of servers $S$ with $|S| \geq n$.
\end{theorem}

\begin{proof}
    By \Cref{lem:ratio-soar}, there exists an algorithm $\calA$ such that
    \begin{align*}
        \costD(\calA; n) \leq
        \begin{cases}
            O(\sqrt{n}) & d = 1\\
            O(\sqrt{n \log n}) & d = 2\\
            O(n^{1 - \frac{1}{d}}) & d \geq 3
        \end{cases}\,.
    \end{align*}
    By \Cref{thm:framework-stocas-advers}, we obtain an algorithm $\calA'$ such that $\costD(\calA'; S, n) \leq \optD(S, n) + \costD(\calA; n)$ for every set of servers $S$ with $|S| \geq n$.
    As a result, the regret of $\calA'$ equals
    \begin{align*}
        \regret(\calA'; S, n)
        = \frac{1}{n} \left( \costD(\calA'; S, n) - \optD(S, n) \right)
        \leq \frac{\costD(\calA; n)}{n}\,,
    \end{align*}
    which concludes the proof.
\end{proof}

We briefly explain the contrast that in unbalanced markets with $d \leq 2$, our algorithm exhibits nearly optimal regret (\Cref{thm:any-dis-regret-ub}) with a suboptimal competitive ratio (\Cref{thm:ratio-ahg-excess}).
Intuitively, our regret analysis indicates that $\costD(\calA; S, n) - \optD(S, n)$ is small, yet we have no control over the magnitude of $\optD(S, n)$ even when $\bbD$ is smooth, which prevents us from achieving a better competitive ratio.
\section{Prediction Setting}

In this section, we focus on balanced markets and the adversarial setting.
We assume that the algorithm is equipped with a prediction of request locations in advance.
More formally, the prediction is given in the form of a set of predicted request locations $R'$ with $|R'| = n$, which is provided before the arrival of any request, and the prediction error is defined as $\eta := \opt(R', R)$.
We reiterate that the prediction error is not known to the algorithm.
Throughout, we slightly abuse notations and use $\opt := \opt(S, R)$ to denote the optimal offline cost.
For an algorithm $\calA$ that leverages predictions, we use $\calA(R')$ to indicate that the prediction provided to $\calA$ is $R'$.

Following the literature on algorithm with predictions~\cite{DBLP:journals/talg/AntoniadisCEPS23,DBLP:conf/sigecom/BergerFGT24}, we define the \emph{consistency} of an algorithm as its competitive ratio when the prediction is perfect, i.e., when $\eta = 0$, and the \emph{robustness} of an algorithm as its worst-case competitive ratio over all predictions.
In other words, if an algorithm $\calA$ is $\alpha_1$-consistent and $\alpha_2$-robust, then for all $S$ and $R$, $\cost(\calA(R'); S, R) \leq \alpha_1 \cdot \opt(S, R)$ when $\opt(R', R) = 0$, and $\cost(\calA(R'); S, R) \leq \alpha_2 \cdot \opt(S, R)$ for every prediction $R'$.
We provide our main theorem in this section as follows, which gives a consistent and robust algorithm.
Moreover, the performance of the algorithm degrades smoothly as the prediction error increases.

\begin{theorem}\label{thm:comp-alg}
    Suppose that $m = n$ and consider the adversarial setting.
    Then, given an algorithm $\calA$ that does not use predictions and is $\alpha$-competitive, there exists an algorithm $\calA'$ that leverages predictions and is $O(\min\{\alpha, 1+ \frac{\eta (\alpha + 1) }{\opt}\})$-competitive.
    In particular, $\calA'$ is $O(1)$-consistent and $O(\alpha)$-robust.
    Moreover, if $\calA$ is deterministic, then $\calA'$ is also deterministic.
\end{theorem}

The proof of \Cref{thm:comp-alg} is deferred to \Cref{sec:pred-non-consi}, and we now apply \Cref{thm:comp-alg} to the existing algorithms in the literature.
Recall that without predictions, there exist deterministic $(2n-1)$-competitive algorithms for general metrics \cite{DBLP:journals/jal/KalyanasundaramP93,DBLP:journals/tcs/KhullerMV94,DBLP:conf/approx/Raghvendra16}, a deterministic $O(\log n)$-competitive algorithm for the line metric \cite{DBLP:conf/compgeom/Raghvendra18}, a randomized $O(\log^2 n)$-competitive algorithm for general metrics \cite{DBLP:journals/algorithmica/BansalBGN14}, and a randomized $O(\log n)$-competitive algorithm for doubling metrics \cite{DBLP:conf/icalp/GuptaL12}.
We derive the following corollary.

\begin{corollary}\label{cor:pred-metric}
    For general metrics, there exist a deterministic $O(\min\{n, 1 + \frac{\eta \cdot n}{\opt})\}$-competitive algorithm and a randomized $O(\min\{\log^2 n, 1 + \frac{\eta \cdot \log^2 n}{\opt} \})$-competitive algorithm.
    For the line metric, there exists a deterministic $O(\min\{\log n, 1 + \frac{\eta \cdot \log n}{\opt})\}$-competitive algorithm.
    For doubling metrics, there exists a randomized $O(\min\{\log n, 1 + \frac{\eta \cdot \log n}{\opt})\}$-competitive algorithm.
\end{corollary}

\subsection{Proof of \Cref{thm:comp-alg}}
\label{sec:pred-non-consi}

We describe our high-level proof strategy as follows.
We first apply \Cref{thm:alg-with-addi-info} to give a consistent algorithm whose performance is optimal with perfect predictions and smoothly depends on the prediction error.
Then, we robustify this algorithm by combining it with the provided algorithm $\calA$, which is by nature robust yet inconsistent.

We start by presenting a consistent algorithm using \Cref{thm:alg-with-addi-info}.

\begin{lemma}\label{lem:alg-nonrobust}
    Given an $\alpha$-competitive algorithm $\calA$ that does not use predictions, there exists an algorithm $\calA'$ that leverages predictions and is $(1 + \frac{\eta (\alpha + 1)}{\opt})$-competitive.
    In particular, $\calA'$ is $1$-consistent.
    Moreover, if $\calA$ is deterministic, then $\calA'$ is also deterministic.
\end{lemma}

\begin{proof}
    We construct $\calA'$ as follows.
    Given as input a set of servers $S$, a tuple of requests $R$, and a prediction of request locations $R'$ with $|S| = |R| = |R'|$ and $\eta = \opt(R, R')$, let $\calA'$ be the algorithm as per \Cref{thm:alg-with-addi-info} by setting $I = R'$.
    It follows that,
    \begin{align*}
        \cost(\calA'(R'); S, R)
        &\leq \opt(S, R') + \cost(\calA; R', R)
        \leq \opt(S, R) + \opt(R, R') + \cost(\calA; R', R)\\
        &\leq \opt(S, R) + \opt(R, R') + \alpha \cdot \opt(R', R)
        = \opt(S, R) + \eta (\alpha + 1),
    \end{align*}
    where the first inequality holds by the guarantee of \Cref{thm:alg-with-addi-info}, the second inequality holds by \Cref{lem:tri-opt}, and the third inequality holds due to the competitive ratio of $\calA$.
    Finally, if $\calA$ is deterministic, then $\calA'$ being deterministic follows from the guarantee of \Cref{thm:alg-with-addi-info}.
\end{proof}

Next, we robustify the consistent algorithm given by \Cref{lem:alg-nonrobust}.
A powerful machinery for deterministically combining several algorithms into one in a black-box manner is proposed by \cite{DBLP:journals/jcss/FiatRR94} for the $k$-server problem and is subsequently applied by \cite{DBLP:journals/jacm/LykourisV21,DBLP:journals/talg/AntoniadisCEPS23} under the context of online algorithm with predictions.
We will use the same machinery to combine our consistent but non-robust algorithm with the given algorithm $\calA$ that does not use predictions, which is by definition robust but inconsistent, leading to a robust and consistent algorithm.

We adapted the framework for combining several online algorithms in \cite{DBLP:journals/talg/AntoniadisCEPS23} to online metric matching, which is a direct consequence of combining \cite[Theorem 18 and Lemma 21]{DBLP:journals/talg/AntoniadisCEPS23}.

\begin{lemma}[Theorem $18$ and Lemma $21$ in \cite{DBLP:journals/talg/AntoniadisCEPS23}]\label{lem:det-comb}
    Given $k$ deterministic algorithms $\calA_1, \ldots, \calA_k$ for online metric matching, for every $1 < \gamma \leq 2$, there exists a deterministic algorithm $\calA^{\com}$ for online metric matching such that for all $S$ and $R$ with $|S| = |R|$,
    \begin{align*}
        \cost(\calA^{\com}; S, R)
        \leq \left( \frac{2\gamma^k}{\gamma - 1} + 1 \right) \cdot \min_{i \in [k]} \{\cost(\calA_i; S, R)\}.
    \end{align*}
\end{lemma}

We generalize \Cref{lem:det-comb} to combine randomized algorithms in the following lemma.

\begin{lemma}\label{lem:ran-comb}
    Given $k$ randomized algorithms $\calA_1, \ldots, \calA_k$ for online metric matching, for every $1 < \gamma \leq 2$, there exists a randomized algorithm $\calA^{\com}$ for online metric matching such that for all $S$ and $R$ with $|S| = |R|$,
    \begin{align*}
        \cost(\calA^{\com}; S, R)
        \leq \left( \frac{2\gamma^k}{\gamma - 1} + 1 \right) \cdot \min_{i \in [k]} \{\cost(\calA_i; S, R)\}.
    \end{align*}
\end{lemma}

\begin{proof}
    For each randomized algorithm $\calA_i$, we use $\cost(\calA_i; S, R \mid B)$ to denote the cost incurred by $\calA_i$ with servers $S$ and requests $R$ conditional on the random bits used by $\calA_i$ being $B = b_1 b_2 \cdots$.
    To construct the algorithm $\calA^{\com}$, we first sample the random bits $B = b_1 b_2 \cdots$ used by $\calA_1, \ldots, A_k$.
    Conditional on $B$, $\calA_1, \ldots, \calA_k$ become deterministic, and hence we can apply \Cref{lem:det-comb} to combine $\calA_1, \ldots, \calA_k$.
    By the guarantee of \Cref{lem:det-comb}, for all $S$ and $R$, it holds that
    \begin{align*}
        \cost(\calA^{\com}; S, R)
        &\leq \Expect_B \left[ \left( \frac{2\gamma^k}{\gamma - 1} + 1 \right) \cdot \min_{i \in [k]} \{\cost(\calA_i; S, R \mid B)\} \right]\\
        &\leq \left( \frac{2\gamma^k}{\gamma - 1} + 1 \right) \cdot \min_{i \in [k]} \{\Expect_B [\cost(\calA_i; S, R \mid B)]\}\\
        &= \left( \frac{2\gamma^k}{\gamma - 1} + 1 \right) \cdot \min_{i \in [k]} \{\cost(\calA_i; S, R)\} \,,
    \end{align*}
    where the second inequality follows from Jensen's inequality.
\end{proof}

Now, we are ready to prove \Cref{thm:comp-alg}.

\begin{proof}[Proof of \Cref{thm:comp-alg}]
    We only prove the statement for deterministic algorithms by using Lemmas~\ref{lem:alg-nonrobust} and~\ref{lem:det-comb}, and the statement for randomized algorithms can be established analogously by applying Lemmas~\ref{lem:alg-nonrobust} and~\ref{lem:ran-comb}.
    Given a deterministic $\alpha$-competitive algorithm $\calA$ that does not use predictions, by \Cref{lem:alg-nonrobust}, there exists a deterministic $(1 + \frac{\eta (\alpha + 1)}{\opt})$-competitive algorithm $\calA'$ that leverages the prediction $R'$.
    By applying \Cref{lem:det-comb} with $\gamma = 2$ to combine $\calA$ and $\calA'$, we obtain a deterministic algorithm $\calA^{\com}$ such that
    \begin{align*}
        \cost(\calA^{\com}; S, R)
        \leq O(\min\{\cost(\calA; S, R), \cost(\calA'(R'); S, R)\})
    \end{align*}
    for all $S$ and $R$ with $|S| = |R| = |R'|$, which implies that $\calA^{\com}$ is $O(\min\{\alpha, 1 + \frac{\eta (\alpha + 1)}{\opt}\})$-competitive by the competitive ratios of $\calA$ and $\calA'$, concluding the proof.
\end{proof}
\section{Adversarial Setting}

In this section, we focus on the adversarial setting without predictions.
This setting has been extensively studied in the past \cite{DBLP:journals/algorithmica/BansalBGN14,DBLP:conf/compgeom/Raghvendra18}, whereas all known guarantees are established in balanced markets.
We show that the competitive ratios in balanced markets of greedy algorithms satisfying certain natural conditions can be carried over to unbalanced markets.
Combined with the algorithms of prior work, we obtain the first competitive ratio guarantees for unbalanced markets and the adversarial setting, which in turn leads to improved guarantees for unbalanced markets and the semi-stochastic setting.
In particular, we show that both the $O(\log^3 n)$-competitive algorithm of \cite{DBLP:conf/soda/MeyersonNP06} for general metrics and the $O(\log n)$-competitive algorithm of \cite{DBLP:conf/icalp/GuptaL12} for doubling metrics retain their competitive ratios in unbalanced markets.
The latter result also improves the $O(\sqrt{n})$ competitive ratio for $d = 1$ in \Cref{thm:ratio-ahg-excess}.

Given the execution of an algorithm $\calA$ on an instance with initial servers $S$, we use $N_{\calA}(r \mid S)$ to denote the set of available servers closest to the request $r$ upon its arrival.
We say that an algorithm $\calA$ is \emph{greedy} if it always (randomly) matches a new request to one of the closest available servers, i.e., it matches $r$ to a server in $N_{\calA}(r \mid S)$.
In particular, every greedy algorithm is completely determined by a (randomized) tie-breaking rule which, when facing multiple closest available servers upon the arrival of a new request $r$, (randomly) decide which server in $N_{\calA}(r \mid S)$ to match $r$ with.
We say that the tie-breaking rule of a greedy algorithm $\calA$ is \emph{local} if it only depends on the set of closest available servers $N_{\calA}(r \mid S)$.
Moreover, the local tie-breaking rule of a greedy algorithm $\calA$ is called \emph{non-increasing} if the probability of selecting each server in $N_{\calA}(r \mid S)$ does not increase after adding additional servers into $N_{\calA}(r \mid S)$, i.e., for all sets of servers $S', S''$ with $S' \subseteq S''$, for every $s \in S'$,
\begin{align*}
    \prob{\calA \text{ matches } r \text{ to } s \mid N_{\calA}(r \mid S) = S'}
    \geq \prob{\calA \text{ matches } r \text{ to } s \mid N_{\calA}(r \mid S) = S''}\,.
\end{align*}

We show that the competitive ratio of every greedy algorithm with a local and non-increasing tie-breaking rule in balanced markets is retained in unbalanced markets.

\begin{theorem}\label{thm:greedy-bal-to-unbal}
    For the adversarial setting, if a greedy algorithm $\calA$ with a local and non-increasing tie-breaking rule is $\alpha$-competitive in balanced markets, then $\calA$ is $\alpha$-competitive in unbalanced markets.
\end{theorem}

We remark here that the resulting competitive ratio $\alpha$ for unbalanced markets can depend on the number of requests but is independent of the number of excess servers.

\begin{proof}[Proof of \Cref{thm:greedy-bal-to-unbal}.]
    Let $\calA$ be a greedy algorithm with a local and non-increasing tie-breaking rule.
    We claim that for all sets of servers $S, S'$ and tuple of requests $R$ with $S' \subseteq S$ and $|S'| \geq |R|$, $\cost(\calA; S, R) \leq \cost(\calA; S', R)$.
    Before proving this claim, we first apply it to show the desired competitive ratio of $\calA$ in unbalanced markets.
    Let $S$ be the set of servers and $R$ be the tuple of requests with $|S| \geq |R| = n$, and let $S' \subseteq S$ satisfying $|S'| = n$ and $\opt(S, R) = \opt(S', R)$, i.e., $S'$ consists of all matched servers in $S$ in the min-cost matching between $S$ and $R$ that matches all requests in $R$.
    Since $\calA$ is $\alpha$-competitive in balanced markets,
    \begin{align*}
        \cost(\calA; S', R) \leq \alpha \cdot \opt(S', R) = \alpha \cdot \opt(S, R)\, .
    \end{align*}
    Given that $S' \subseteq S$ and $|S'| \geq |R|$, by the claim,
    \begin{align*}
        \cost(\calA; S, R)
        \leq \cost(\calA; S', R)
        \leq \alpha \cdot \opt(S, R)\,,
    \end{align*}
    implying that $\calA$ is $\alpha$-competitive in unbalanced markets.

    It remains to prove the claim, which we establish by induction on $|R|$.
    When $|R| = 0$, we have $\cost(\calA; S, R) = \cost(\calA; S', R) = 0$, and the statement straightforwardly holds.
    Assume for induction that the statement holds for $|R| = k$ with $k \geq 0$, and we show that the statement also holds for $|R| = k + 1$.
    Given $S, S', R$ with $S' \subseteq S$ and $|S'| \geq |R| = k + 1$, let $r \in R$ be the first arrived request.
    We use $M_{\calA}(r\mid S)$ and $M_{\calA}(r\mid S')$ to denote the (random) server that $\calA$ matches $r$ to when the set of initial servers is $S$ and $S'$, respectively.
    Since $\calA$ is greedy, $M_{\calA}(r \mid S) \in N_{\calA}(r \mid S)$ and $M_{\calA}(r \mid S') \in N_{\calA}(r \mid S')$.
    The following lemma asserts that either $N_{\calA}(r \mid S) \cap N_{\calA}(r \mid S') = \emptyset$ or $N_{\calA}(r \mid S') \subseteq N_{\calA}(r \mid S)$.

    \begin{lemma}\label{lem:case-empty-include}
        One of the following must hold: $N_{\calA}(r \mid S) \cap N_{\calA}(r \mid S') = \emptyset$ or $N_{\calA}(r \mid S') \subseteq N_{\calA}(r \mid S)$.
    \end{lemma}

    \begin{proof}
        Suppose for contradiction that $N_{\calA}(r \mid S) \cap N_{\calA}(r \mid S') \neq \emptyset$ and $N_{\calA}(r \mid S') \not\subseteq N_{\calA}(r \mid S)$, and let $s \in N_{\calA}(r \mid S) \cap N_{\calA}(r \mid S')$ and $s' \in N_{\calA}(r \mid S') \setminus N_{\calA}(r \mid S)$.
        Since $r$ is the first arrived request, all servers in $S$ (resp. $S'$) are available upon the arrival of $r$, and hence $N_{\calA}(r \mid S)$ (resp. $N_{\calA}(r \mid S')$) contains all servers in $S$ (resp. $S'$) closest to $r$.
        Moreover, note that $s, s' \in N_{\calA}(r \mid S') \subseteq S' \subseteq S$.
        On one hand, since $s, s' \in N_{\calA}(r \mid S')$, $\delta(s, r) = \delta(s', r)$.
        On the other hand, since $s \in N_{\calA}(r \mid S)$ and $s' \notin N_{\calA}(r \mid S)$, $\delta(s, r) < \delta(s', r)$, leading to a contradiction.
    \end{proof}
    
    We analyze the expected cost of $\calA$ under both cases in \Cref{lem:case-empty-include} separately.
    We start with the easier case where $N_{\calA}(r \mid S) \cap N_{\calA}(r \mid S') = \emptyset$, which implies that $N_{\calA}(r \mid S) \cap S' = \emptyset$ and $\delta(r, s) < \delta(r, s')$ for all $s \in N_{\calA}(r \mid S)$ and $s' \in N_{\calA}(r \mid S')$.
    By $N_{\calA}(r \mid S) \cap S' = \emptyset$, we know that $M_{\calA}(r \mid S) \notin S'$ and hence $S' \setminus \{M_{\calA}(r \mid S')\} \subseteq S' \subseteq S \setminus \{M_{\calA}(r \mid S)\}$.
    By the inductive hypothesis, $\cost(\calA; S \setminus \{M_{\calA}(r \mid S)\}, R \setminus \{r\}) \leq \cost(\calA; S' \setminus \{M_{\calA}(r \mid S')\}, R \setminus \{r\})$.
    As a result,
    \begin{align*}
        \cost(\calA; S, R)
        &= \Expect[\delta(r, M_{\calA}(r \mid S)) + \cost(\calA; S \setminus \{M_{\calA}(r \mid S)\}, R \setminus \{r\})]\\
        &< \Expect[\delta(r, M_{\calA}(r \mid S')) + \cost(\calA; S' \setminus \{M_{\calA}(r \mid S')\}, R \setminus \{r\})]\\
        &= \cost(\calA; S', R)\,,
    \end{align*}
    as desired.

    Next, we consider the case where $N_{\calA}(r \mid S') \subseteq N_{\calA}(r \mid S)$, which implies that $\delta(r, s) = \delta(r, s')$ for all $s, s' \in N_{\calA}(r \mid S)$.
    For every $s \in N_{\calA}(r \mid S)$ and $s' \in N_{\calA}(r \mid S')$, define
    \begin{align*}
        P_S(s) &:= \prob{M_{\calA}(r \mid S) = s}, & Q_S(s) &:= \cost(\calA; S \setminus \{s\}, R \setminus \{r\}) \, ,\\
        P_{S'}(s') &:= \prob{M_{\calA}(r \mid S') = s'}, & Q_{S'}(s') &:= \cost(\calA; S' \setminus \{s'\}, R \setminus \{r\}) \, .
    \end{align*}
    Since the tie-breaking rule of $\calA$ is local and non-increasing, $P_{S'}(s) \geq P_S(s)$ for every $s \in N_{\calA}(r \mid S')$.
    As a result, 
    \begin{align*}
        &\cost(\calA; S, R) - \cost(\calA; S', R)\\
        & = \sum_{s \in N_{\calA}(r \mid S)} P_S(s) \cdot (\delta(r, s) + Q_S(s)) - \sum_{s \in N_{\calA}(r \mid S')} P_{S'}(s) \cdot (\delta(r, s) + Q_{S'}(s))\\
        & = \sum_{s \in N_{\calA}(r \mid S)} P_S(s) \cdot Q_S(s) - \sum_{s \in N_{\calA}(r \mid S')} P_{S'}(s) \cdot Q_{S'}(s)\\
        & = \sum_{s \in N_{\calA}(r \mid S')} (P_S(s) \cdot Q_S(s) - P_{S'}(s) \cdot Q_{S'}(s)) + \sum_{s \in N_{\calA}(r \mid S) \setminus N_{\calA}(r \mid S')} P_S(s) \cdot Q_S(s)\\
        & \leq \sum_{s \in N_{\calA}(r \mid S')} Q_{S'}(s) \cdot (P_S(s) - P_{S'}(s)) + \sum_{s \in N_{\calA}(r \mid S) \setminus N_{\calA}(r \mid S')} P_S(s) \cdot Q_S(s) \,,
    \end{align*}
    where the second equality holds since $\delta(r, s) = \delta(r, s')$ for all $s, s' \in N_{\calA}(r \mid S)$, and the inequality holds since $Q_S(s) \leq Q_{S'}(s)$ for every $s \in N_{\calA}(r \mid S')$ by the inductive hypothesis.
    Note that for all $s \in N_{\calA}(r \mid S) \setminus N_{\calA}(r \mid S')$ and $s' \in N_{\calA}(r \mid S')$, we have $S' \setminus \{s'\} \subseteq S' \subseteq S \setminus \{s\}$ and hence $Q_S(s) \leq Q_{S'}(s')$ by the inductive hypothesis.
    Let $B := \min_{s' \in N_{\calA}(r \mid S')} Q_{S'}(s')$, and it follows that $Q_S(s) \leq B \leq Q_{S'}(s')$ for all $s \in N_{\calA}(r \mid S) \setminus N_{\calA}(r \mid S')$ and $s' \in N_{\calA}(r \mid S')$.
    Therefore,
    \begin{align*}
        &\cost(\calA; S, R) - \cost(\calA; S', R) \\
        &\leq \sum_{s \in N_{\calA}(r \mid S')} Q_{S'}(s) \cdot (P_S(s) - P_{S'}(s)) + \sum_{s \in N_{\calA}(r \mid S) \setminus N_{\calA}(r \mid S')} P_S(s) \cdot Q_S(s)\\
        &\leq \sum_{s \in N_{\calA}(r \mid S')} B \cdot (P_S(s) - P_{S'}(s)) + \sum_{s \in N_{\calA}(r \mid S) \setminus N_{\calA}(r \mid S')} P_S(s) \cdot B\\
        &= B \cdot \left( \sum_{s \in N_{\calA}(r \mid S)} P_S(s) - \sum_{s \in N_{\calA}(r \mid S')} P_{S'}(s) \right)\\
        &= 0\,,
    \end{align*}
    where the second inequality holds since $P_{S'}(s) \geq P_S(s)$ for every $s \in N_{\calA}(r \mid S')$.
    This concludes the proof.
\end{proof}

\subsection{Applications to Existing Algorithms}
\label{sec:apply-hst}

Now, we apply \Cref{thm:greedy-bal-to-unbal} to algorithms of prior work.
Recall that for balanced markets and the adversarial setting, the algorithm of \cite{DBLP:conf/soda/MeyersonNP06}, abbreviated as the \emph{MNP} algorithm, is $O(\log^3 n)$-competitive for general metrics, and the \emph{Random-Subtree} algorithm of \cite{DBLP:conf/icalp/GuptaL12} is $O(\log n)$-competitive for doubling metrics.
We show that the competitive ratios of both algorithms are retained in unbalanced markets.

Both the MNP and Random Subtree algorithms are designed based on the same framework.
They first randomly embed the original metric into a \emph{Hierarchically well-Separated Tree (HST)} metric, which preserves the distance between every pair of locations up to an $O(\log n)$ distortion.
Next, they show that a randomized greedy subroutine with a specific tie-breaking rule is $\beta$-competitive for the HST metric.
Overall, this leads to an $O(\beta \log n)$-competitive algorithm for the original metric.
For our purpose, it suffices to show that the tie-breaking rule of the greedy subroutine for the HST metric is local and non-increasing.

Recall that an HST is an edge-weighted rooted tree satisfying the following property: For each node $x$ of the HST, the distances from $x$ to all leaf nodes in the subtree rooted at $x$ are identical.\footnote{The definition of HSTs also includes other conditions, but they are not needed for us.}
Moreover, the metric embedding phase of the algorithms will map each location in the original metric to a leaf node of the HST.
From now on, we will restrict our attention to the HST metric with all servers and requests located at the HST's leaf nodes.

We start by analyzing the greedy subroutine in the MNP algorithm, which works as follows.
Upon the arrival of a new request $r$, assuming the set of initial servers to be $S$, denote $A(r \mid S)$ as the lowest ancestor of $r$ the subtree rooted at whom contains at least one available server, and the subroutine matches $r$ to a uniformly random available server in the subtree rooted at $A(r \mid S)$.
In the following lemma, we show that this tie-breaking rule is local and non-increasing.

\begin{lemma}\label{lem:tie-BBGN}
    The tie-breaking rule of the greedy subroutine in the MNP algorithm is local and non-increasing.
\end{lemma}

The proof of \Cref{lem:tie-BBGN} is deferred to \Cref{sec:proof-tie-BBGN}.
The following corollary is a direct consequence of combining \Cref{thm:greedy-bal-to-unbal} and \Cref{lem:tie-BBGN} with the competitive ratio of the MNP algorithm in balanced markets~\cite{DBLP:conf/soda/MeyersonNP06}.

\begin{corollary}\label{cor:adv-arbi-metric}
    In the adversarial setting and unbalanced markets, there exists an $O(\log^3 n)$-competitive algorithm for general metrics.
\end{corollary}

We remark here that in the adversarial setting and unbalanced markets, an $O(\log^2 n)$-competitive algorithm for general metrics is known~\cite{DBLP:journals/algorithmica/BansalBGN14}.
However, the subroutine of this algorithm is not greedy, and hence \Cref{thm:greedy-bal-to-unbal} does not apply.

Next, we analyze the greedy subroutine in the Random-Subtree algorithm, which works as follows.
Upon the arrival of a new request $r$, starting off at $A(r \mid S)$, the subroutine repeatedly moves to a uniformly random subtree of the current node that contains at least one available server, until a leaf node is reached.
Then, $r$ is matched to a uniformly random available server at this leaf node.
We show in the following lemma that this tie-breaking rule is local and non-increasing.

\begin{lemma}\label{lem:random-subtree-tie}
    The tie-breaking rule of the greedy subroutine in the Random Subtree algorithm is local and non-increasing.
\end{lemma}

The proof of \Cref{lem:random-subtree-tie} is deferred to \Cref{sec:proof-random-subtree-tie}.
Combining \Cref{thm:greedy-bal-to-unbal} and \Cref{lem:random-subtree-tie} with the competitive ratio of the Random Subtree algorithm in balanced markets~\cite{DBLP:conf/icalp/GuptaL12}, we obtain the following corollary.

\begin{corollary}\label{cor:adv-unbal-line}
    In the adversarial setting and unbalanced markets, there exists an $O(\log n)$-competitive algorithm for doubling metrics.
\end{corollary}
\section{Discussion and Future Directions}

In this paper, we present a generic framework for online metric matching to leverage the provided prediction of request locations, and we apply this framework to achieve new results in the semi-stochastic and prediction settings.

For the semi-stochastic setting and balanced markets, we give an $O(1)$-competitive algorithm for the Euclidean metric with smooth distributions and $d \neq 2$.
Nevertheless, it still remains open to give an $O(1)$-competitive algorithm for general metrics and distributions.
The potential approaches for accomplishing this goal comprise (1) improving the analysis of the $O((\log \log \log n)^2)$-competitive algorithm by \cite{DBLP:conf/icalp/GuptaGPW19}, or (2) applying \Cref{thm:framework-stocas-advers} and the SOAR algorithm to the simplified model given by \cite{DBLP:conf/icalp/GuptaGPW19}, where $\bbD$ is uniformly distributed over server locations.
The main technical barrier underlying the latter approach lies in characterizing $\optD(n)$ for general metrics.
In \Cref{sec:tree}, we further substantiate the potential of this approach by showing that it gives an $O(1)$-competitive algorithm for tree metrics and general distributions, recovering the result by \cite{DBLP:conf/icalp/GuptaGPW19} via the first approach.

Besides, for the semi-stochastic setting and unbalanced markets, we derive an algorithm with improved competitive ratios for the Euclidean metric with smooth distributions.
The resulting competitive ratio is optimal for $d \geq 3$, leaving the improvement for $d \leq 2$ an intriguing future direction.
Given that the line metric, i.e., $d = 1$, is typically the most tractable case \cite{DBLP:conf/icalp/GuptaGPW19,DBLP:conf/sigecom/BalkanskiFP23}, it is surprising that it becomes the most intractable with the presence of excess servers, rendering its resolution more appealing.

Moreover, for the adversarial setting without predictions, we show that the competitive ratio of every greedy algorithm with a local and non-increasing tie-breaking rule in balanced markets is preserved in unbalanced markets.
However, our characterization fails to capture all existing algorithms.
A natural subsequent step would be to generalize the performance of other algorithms, such as the algorithms by \cite{DBLP:journals/algorithmica/BansalBGN14,DBLP:conf/approx/Raghvendra16}, to unbalanced markets.

Several research avenues arise from the modeling assumptions in this paper. First, our results assume advance knowledge of $n$, raising the question of whether competitive guarantees are achievable without this information.
Second, the competitive ratios achieved by our algorithm in the semi-stochastic setting for unbalanced markets (\Cref{thm:ratio-ahg-excess}) depend polynomially on $\kappa := m / n$, and hence no longer remain constants when $\kappa = \omega(1)$.
Hence, it would be intriguing to achieve a constant competitive ratio in the regime $m = \omega(n)$.
Moreover, relaxing the assumptions that servers outnumber requests and all requests must be served gives rise to an interesting variant where request rejection is permitted—--a setting that bears similarities to the \emph{multi-unit prophet inequality} problem \cite{DBLP:conf/soda/JiangMZ22}.
Finally, it is appealing to relax the triangle inequality satisfied by the cost function; in \Cref{sec:general-cost}, we make progress toward this direction by generalizing some of our results to the case where the cost function only approximately satisfies the triangle inequality.

\bibliographystyle{alpha}
\bibliography{references}

\clearpage
\appendix

\section{Postponed Proofs in \Cref{sec:optoff-euc}}
\label{sec:proof-optoff-euc}

\subsection{Proof of \Cref{lem:optoff-smooth-lb}}
\label{sec:proof-lem-optoff-smooth-lb}

For $d = 1$, let $F(x) := \int_0^x \mu_{\mathbb{D}}(x) \mathrm{d}x$ for $x \in [0, 1]$ be the cumulative density function of $\mathbb{D}$.
    Let $S = \{s_1, \ldots, s_n\} \sim \mathbb{D}^n$ and $R = \{r_1, \ldots, r_n\} \sim \mathbb{D}^n$ be two sets of random points independently drawn from $\mathbb{D}$.
    Define $P_S := \{F(s_1), \ldots, F(s_n)\}$ and $P_R := \{F(r_1), \ldots, F(r_n)\}$, and each point in $P_S \cup P_R$ follows the uniform distribution over $[0, 1]$.
    Since the expected cost of the optimal matching between two sets of $n$ uniformly random points over $[0, 1]$ is $\Theta(\sqrt{n})$, we have
    \begin{align}\label{eqn:opt-uni-1d}
        \Expect_{S \sim \bbD^n, R \sim \mathbb{D}^n} \left[ \opt(P_S, P_R) \right]
        = \Theta\left(\sqrt{n}\right)\,.
    \end{align}
    Note that for all $s, r \in [0, 1]$, by the $\beta$-smoothness of $\bbD$,
    \begin{align}\label{eqn:mapping-dis}
        |F(s) - F(r)|
        = \left|\int_r^s f(x) \mathrm{d}x\right|
        \leq \left| \int_r^s \beta \mathrm{d}x \right|
        = \beta \cdot |s - r|\,.
    \end{align}
    Let $M$ be the perfect matching between $S$ and $R$ with the minimum cost, and let $M'$ be the perfect matching between $P_S$ and $P_R$ induced by $M$ such that $F(s) \in P_S$ is matched with $F(M(s)) \in P_R$ for each $s \in S$.
    It follows that
    \begin{align}\label{eqn:cost-matching-xy}
        \sum_{s \in S} |F(s) - F(M(s))|
        \geq \opt(P_S, P_R)\,.
    \end{align}
    As a result,
    \begin{align*}
        \opt(S, R)
        = \sum_{s \in S} |s - M(s)|
        \geq \sum_{s \in S} \frac{1}{\beta} \cdot |F(s) - F(M(s))|
        \geq \frac{1}{\beta} \cdot \opt(P_S, P_R)\,,
    \end{align*}
    where the equality holds by the definition of $M$, the first inequality holds by \eqref{eqn:mapping-dis}, and the second inequality holds by \eqref{eqn:cost-matching-xy}.
    Taking expectations on both sides, we obtain
    \begin{align*}
        \opt_{\mathbb{D}}(n)
        = \Expect_{S \sim \bbD^n, R \sim \mathbb{D}^n} \left[ \opt(S, R) \right]
        \geq \frac{1}{\beta} \cdot \Expect_{S \sim \bbD^n, R \sim \mathbb{D}^n} \left[ \opt(P_S, P_R) \right]
        \geq \Omega\left(\sqrt{n} / \beta\right) \,,
    \end{align*}
    where the last inequality holds by \eqref{eqn:opt-uni-1d}.
    This concludes the lower bound for $d = 1$.

The lower bound for $d \geq 2$ is proved by employing nearest-neighbor-distance, which is provided in the following lemma.

\begin{lemma}\label{lem:near-neigh-dis-smooth}
    Assume $\mathbb{D}$ to be a $\beta$-smooth distribution over $[0, 1]^d$ for $\beta \geq 1$.
    Then, for all $d \geq 1$ and $S \subseteq [0, 1]^d$,
    \begin{align*}
        \Expect_{r \sim \mathbb{D}} \left[ \min_{s \in S} \norm{s - r}_2 \right] \geq \Omega\left(\beta^{-\frac{1}{d}} |S|^{-\frac{1}{d}}\right)\,.
    \end{align*}
\end{lemma}

\begin{proof}
    Fix $S \subseteq [0, 1]^d$.
    For all $s \in S$ and $\ell \geq 0$, let $B(s, \ell) := \{x \in [0, 1]^d \mid \norm{x - s}_2 \leq \ell\}$ denote the ball restricted in $[0, 1]^d$ with center $s$ and radius $\ell$, and, by the $\beta$-smoothness of $\bbD$,
    \begin{align*}
        \Prob_{r \sim \mathbb{D}} [\norm{s - r}_2 \leq \ell]
        = \int_{B(s, \ell)} \mu_{\mathbb{D}}(x) \mathrm{d}x
        \leq \int_{B(s, \ell)} \beta \mathrm{d}x
        \leq C \beta \cdot \ell^d
    \end{align*}
    for some constant $C > 0$.
    For all $s \in S$ and $r \in [0, 1]^d$, let $O_{s, r}$ denote the event that $s \in \argmin_{s' \in S} \norm{s' - r}_2$, and let $p_s := \Prob_{r \sim \mathbb{D}} [O_{s, r}]$.
    For all $s \in S$ and $0 < \ell \leq (p_s/ C \beta)^{1 / d}$, it holds that
    \begin{align*}
        \Prob_{r \sim \mathbb{D}} \left[ O_{s, r} \land (\norm{s - r}_2 \geq \ell) \right]
        &= \Prob_{r \sim \mathbb{D}} \left[ O_{s, r} \right] - \Prob_{r \sim \mathbb{D}} \left[ O_{s, r} \land (\norm{s - r}_2 \leq \ell) \right]\\
        &\geq p_s - \Prob_{r \sim \mathbb{D}} \left[ \norm{s - r}_2 \leq \ell \right]\\
        &\geq p_s - C \beta \cdot \ell^d\,.
    \end{align*}
    As a result, for every $s \in S$,
    \begin{align*}
        \Expect_{r \sim \mathbb{D}} \left[ \norm{s - r}_2 \mid O_{s, r} \right]
        &\geq \int_0^{(p_s / C \beta)^{1 / d}} \Prob_{r \sim \mathbb{D}} \left[\norm{s - r}_2 \geq \ell \mid O_{s, r} \right] \mathrm{d} \ell\\
        &= \frac{1}{p_s} \int_0^{(p_s / C \beta)^{1 / d}} \Prob_{r \sim \mathbb{D}} [O_{s, r} \land (\norm{s - r}_2 \geq \ell)] \mathrm{d} \ell\\
        &\geq \frac{1}{p_s} \int_0^{(p_s / C \beta)^{1 / d}} (p_s - C \beta \cdot \ell^d) \mathrm{d} \ell\\
        &= \frac{d}{d + 1} \cdot \left( \frac{p_s}{C \beta} \right)^{1 / d}\,.
    \end{align*}
    Therefore,
    \begin{align*}
        \Expect_{r \sim \mathbb{D}} \left[ \min_{s \in S} \norm{s - r}_2 \right]
        &= \sum_{s \in S} \Prob_{r \sim \mathbb{D}}[O_{s, r}] \cdot \Expect_{r \sim \mathbb{D}} [\norm{s - r}_2 \mid O_{s, r}] \\
        &\geq \sum_{s \in S} \Omega\left(\beta^{-\frac{1}{d}} p_s^{1 + \frac{1}{d}} \right)
        \geq \Omega\left(\beta^{-\frac{1}{d}} |S|^{-\frac{1}{d}}\right)\,,
    \end{align*}
    where the last inequality holds by the convexity of $f(x) = x^{1 + \frac{1}{d}}$ for $x \geq 0$.
\end{proof}

    Now, we are ready to prove the lower bound for $d \geq 2$.
    By \Cref{lem:near-neigh-dis-smooth},
    \begin{align*}
        \opt_{\mathbb{D}}(n)
        &= \Expect_{S \sim \bbD^n, R \sim \mathbb{D}^n} \left[ \opt(S, R) \right]
        = \Expect_{S \sim \mathbb{D}^n} \left[ \Expect_{R \sim \mathbb{D}^n} \left[ \opt(S, R) \right] \right]\\
        &\geq \Expect_{S \sim \mathbb{D}^n} \left[n \cdot \Expect_{r \sim \mathbb{D}} \left[ \min_{s \in S} \norm{s - r}_2 \right] \right]
        \geq \Omega\left(\beta^{-\frac{1}{d}} n^{1 - \frac{1}{d}}\right)\,,
    \end{align*}
    which concludes the proof.

\subsection{Proof of \Cref{lem:lb-optoff-excess}}
\label{sec:proof-lem-lb-optoff-excess}

By \Cref{lem:near-neigh-dis-smooth} and the assumption that $|S| = \kappa n$,
    \begin{align*}
        \opt_{\mathbb{D}}(S, n)
        = \Expect_{R \sim \mathbb{D}^n} \left[ \opt(S, R) \right]
        \geq n \cdot \Expect_{r \sim \mathbb{D}} \left[ \min_{s \in S} \norm{s - r}_2 \right]
        \geq \Omega\left((\beta\kappa)^{-\frac{1}{d}} n^{1 - \frac{1}{d}}\right) \,,
    \end{align*}
    as desired.
\section{Postponed Proofs in \Cref{sec:apply-hst}}
\label{sec:proof-greedy}

\subsection{Proof of \Cref{lem:tie-BBGN}}
\label{sec:proof-tie-BBGN}

    We use $\calA$ to denote the greedy subroutine in the MNP algorithm.
    The tie-breaking rule of $\calA$ is local by definition, and it remains to show that it is non-increasing.
    Recall that $N_{\calA}(r \mid S)$ denotes the set of available servers closest to the request $r$ upon its arrival when the set of initial servers is $S$.
    For all sets of servers $S', S''$ with $S' \subseteq S''$, for every $s \in S'$,
    \begin{align*}
        \prob{\calA \text{ matches } r \text{ to } s \mid N_{\calA}(r \mid S) = S'}
        &= \frac{1}{|S'|} \,,\\
        \prob{\calA \text{ matches } r \text{ to } s \mid N_{\calA}(r \mid S) = S''}
        &= \frac{1}{|S''|} \,. 
    \end{align*}
    Therefore,
    \begin{align*}
        \prob{\calA \text{ matches } r \text{ to } s \mid N_{\calA}(r \mid S) = S'}
        \geq \prob{\calA \text{ matches } r \text{ to } s \mid N_{\calA}(r \mid S) = S''} \,,
    \end{align*}
    implying that the tie-breaking rule of $\calA$ is non-increasing.

\subsection{Proof of \Cref{lem:random-subtree-tie}}
\label{sec:proof-random-subtree-tie}

    We use $\calA$ to denote the greedy subroutine in the Random Subtree algorithm.
    The tie-breaking rule of $\calA$ is local by definition, and it remains to show that it is non-increasing.
    For all leaf node $x$ of the HST and $\ell \geq 0$, denote $A_{\ell}(x)$ as the $\ell$-th ancestor of $x$ and $T_{\ell}(x)$ as the subtree rooted at $A_{\ell}(x)$.
    That is, $A_0(x)$ is $x$, $A_1(x)$ is the parent of $x$, and so on.
    Fix a request $r$.
    For all sets of servers $S', S''$ with $S' \subseteq S''$, let $\ell'$ satisfy $A_{\ell'}(r) = A(r \mid S)$ when $N_{\calA}(r \mid S) = S'$, and let $\ell''$ satisfy $A_{\ell''}(r) = A(r \mid S)$ when $N_{\calA}(r \mid S) = S''$.
    Note that $\ell'$ and $\ell''$ only depend on $S'$ and $S''$, respectively.

    We first show that $\ell' = \ell''$.
    Assume for contradiction that $\ell' \neq \ell''$. 
    We only analyze the case where $\ell' < \ell''$, and the case where $\ell' > \ell''$ can be handled analogously.
    When $\ell' < \ell''$, all servers in $S'$ are located in $T_{\ell'}(r)$, and all servers in $S''$ are located in $T_{\ell''}(r)$.
    Moreover, by the definition of $\ell''$, no server in $S''$ is located in $T_{\ell'' - 1}(r)$.
    Since $T_{\ell'}(r)$ is a subtree in $T_{\ell'' - 1}(r)$, we have $S' \cap S'' = \emptyset$, contradicting the assumption that $S' \subseteq S''$.

    Fix a server $s \in S'$.
    For every $i \in [\ell']$, denote $w_i'$ as the number of subtrees of $A_i(s)$ containing at least one available server upon the arrival of $r$ when $N_{\calA}(r \mid S) = S'$.
    Moreover, define $w_0'$ as the number of available servers collocated with $s$ (including $s$) upon the arrival of $r$ when $N_{\calA}(r \mid S) = S'$.
    Recall that $\ell' = \ell''$, and define $w_i''$ for every $i \in \{0, 1, \ldots, \ell'\}$ similarly when $N_{\calA}(r \mid S) = S''$.
    By the tie-breaking rule of $\calA$,
    \begin{align*}
        \prob{\calA \text{ matches } r \text{ to } s \mid N_{\calA}(r \mid S) = S'}
        &= \prod_{i=0}^{\ell'} \frac{1}{w_i'} \, ,\\
        \prob{\calA \text{ matches } r \text{ to } s \mid N_{\calA}(r \mid S) = S''}
        &= \prod_{i=0}^{\ell'} \frac{1}{w_i''} \,.
    \end{align*}
    Since $S' \subseteq S''$, we have $w_i' \leq w_i''$ for every $i \in \{0, 1, \ldots, \ell'\}$.
    Therefore,
    \begin{align*}
        \prob{\calA \text{ matches } r \text{ to } s \mid N_{\calA}(r \mid S) = S' }
        \geq \prob{\calA \text{ matches } r \text{ to } s \mid N_{\calA}(r \mid S) = S''}\,,
    \end{align*}
    implying that the tie-breaking rule of $\calA$ is non-increasing.

\section{Semi-Stochastic Setting: An Alternative $O(1)$-Competitive Algorithm for Tree Metrics with General Distributions}
\label{sec:tree}

In this section, we consider balanced markets and tree metrics with general distributions.
In particular, a tree metric is defined by the shortest path distances in a (weighted) tree $T = (V, E, w)$.
We will prove that the SOAR algorithm given in \Cref{lem:matching-cost-soar} is $O(1)$-competitive in the stochastic setting, which leads to an $O(1)$-competitive algorithm in the semi-stochastic setting by \Cref{thm:framework-stocas-advers}, recovering the result by \cite{DBLP:conf/icalp/GuptaGPW19}.

\begin{theorem}\label{thm:uni-metric-adv}
    Suppose that $m = n$.
    Then, there exists an $O(1)$-competitive algorithm for the semi-stochastic setting and tree metrics with general distributions.
\end{theorem}

To prove \prettyref{thm:uni-metric-adv}, we utilize the result by \cite{DBLP:conf/icalp/GuptaGPW19}, which holds for any metric, that only a constant factor in the competitive ratio is lost by assuming that $\bbD$ is the uniform distribution over server locations.
That is, for every server $s \in S$, $\Prob_{r \sim \bbD}[r = s] = 1/n$.

\begin{lemma}[Lemma 2.1 in \cite{DBLP:conf/icalp/GuptaGPW19}]\label{lem:nonuni-dis}
    Suppose that $m = n$ and consider the semi-stochastic setting.
    Given an $\alpha$-competitive algorithm for the uniform distribution over server locations, we can construct a $(2\alpha + 1)$-competitive algorithm for general distributions.
\end{lemma}

Let $T = (V, E, w)$ be the tree that induces the tree metric.
Given \prettyref{lem:nonuni-dis}, it suffices to consider distribution $\mathbb{D}$ of the following form: Each node $v \in V$ is associated with an integer $a_v \geq 0$ such that $\sum_{v \in V} a_v = n$ and $\Prob_{r \sim \mathbb{D}}[r = v] = a_v / n$.
For each edge $e \in E$, let $T_1(e)$ and $T_2(e)$ be two subtrees of $T$ obtained by deleting $e$, and define $n_e := \sum_{v \in T_1(e)} a_v$.
We give a tight characterization of $\optD(t)$ in the following lemma.

\begin{lemma}\label{lem:uni-metric-optoff}
    For each $t \in [n]$, $\optD(t) = \Theta \left( \frac{\sqrt{t}}{n} \sum_{e \in E} w_e \sqrt{n_e (n - n_e)} \right)$.
\end{lemma}

\begin{proof}
    Fix $t \in [n]$.
    For each $e \in E$, let $P_e \sim \Binom(t, n_e / n)$ and $Q_e \sim \Binom(t, n_e / n)$ be random variables denoting the number of servers located in $T_1(e)$ and the number of requests located in $T_1(e)$, respectively.
    Given a realization of $(P_e, Q_e)_{e \in E}$, in the optimal perfect matching between servers and requests, for each edge $e \in E$, exactly $|P_e - Q_e|$ matched server-request pairs will utilize $e$ in their paths, incurring a cost of $w_e \cdot |P_e - Q_e|$.
    Hence,
    \begin{align*}
        \optD(t)
        = \Expect \left[ \sum_{e \in E} w_e \cdot |P_e - Q_e| \right]
        = \sum_{e \in E} w_e \cdot \Expect [|P_e - Q_e|]\,.
    \end{align*}
    
    The lower bound for $\optD(t)$ relies on the following probabilistic bound.
    
\begin{lemma}[\cite{berend2013sharp}]\label{lem:binom-pro-bound}
        Let $Z \sim \Binom(n, p)$, with $n \geq 2$ and $p \in [1 / n, 1 - 1 / n]$.
        Then,
        \begin{align*}
            \Expect[|Z - \Expect Z|]
            \geq \std(Z) / \sqrt{2}\,,
        \end{align*}
        where $\std(Z)$ denotes the standard deviation of $Z$.
    \end{lemma}

    We lower bound $\optD(t)$ by
    \begin{align*}
        \optD(t)
        &= \sum_{e \in E} w_e \cdot \Expect[|P_e - Q_e|]
        \geq \sum_{e \in E} w_e \cdot \Expect [|P_e - \Expect P_e|]\\
        &\geq \frac{1}{\sqrt{2}} \sum_{e \in E} w_e \cdot \std(P_e)
        = \Omega \left( \frac{\sqrt{t}}{n} \sum_{e \in E} w_e \sqrt{n_e (n - n_e)} \right)\,,
    \end{align*}
    where the first inequality holds by Jensen's inequality, and the second inequality holds by \Cref{lem:binom-pro-bound}.
    Moreover, $\optD(t)$ is upper bounded by
    \begin{align*}
        \optD(t)
        &= \sum_{e \in E} w_e \cdot \Expect[|P_e - Q_e|]
        \leq \sum_{e \in E}w_e \left( \Expect \left[\left|P_e - \Expect P_e \right | \right] + \Expect \left[ \left |Q_e - \Expect Q_e \right| \right] \right)\\
        &\leq \sum_{e \in E} w_e (\std(P_e) + \std(Q_e))
        = O \left( \frac{\sqrt{t}}{n} \sum_{e \in E} w_e \sqrt{n_e (n - n_e)} \right)\,,
    \end{align*}
    where the first inequality holds since $\Expect P_e = \Expect Q_e$ for every $e \in E$, and the second inequality holds by Jensen's inequality.
\end{proof}

Now, we are ready to prove \Cref{thm:uni-metric-adv}.

\begin{proof}[Proof of \Cref{thm:uni-metric-adv}]
    We assume $\bbD$ to be the uniform distribution over server locations, which only causes an additional constant factor in the competitive ratio by \Cref{lem:nonuni-dis}.
    In other words, let $a_v$ be the number of servers located at $v \in V$, and it holds that $\Prob_{r \sim \mathbb{D}} [r = v] = a_v / n$.
    By \Cref{thm:framework-stocas-advers}, it suffices to give an $O(1)$-competitive algorithm in the stochastic setting.

    Let $\calA$ be the algorithm as per \Cref{lem:matching-cost-soar}.
    By \Cref{lem:uni-metric-optoff},
    \begin{align*}
        \costD(\calA; n)
        &= \sum_{t=1}^n \frac{1}{t} \cdot \optD(t) \\
        &= O \left( \sum_{t=1}^n \frac{1}{n\sqrt{t}} \sum_{e \in E} w_e \sqrt{n_e(n - n_e)} \right)\\
        &= O \left( \frac{1}{n} \sum_{e \in E} w_e \sqrt{n_e (n - n_e)} \sum_{t=1}^n \frac{1}{\sqrt{t}} \right) \\
        &= O \left( \frac{1}{\sqrt{n}} \sum_{e \in E} w_e \sqrt{n_e (n - n_e)} \right) \\
        &= O\left( \optD(n) \right)\,,
    \end{align*}
    concluding that $\calA$ is $O(1)$-competitive in the stochastic setting.
\end{proof}
\section{Semi-Stochastic Setting: Beyond Triangle Inequality}
\label{sec:general-cost}

In the remainder of the paper, we consider cost functions that satisfy the triangle inequality.
In this section, we generalize some of our results to the case where the cost function only approximately satisfies the triangle inequality.
Specifically, we say that a cost function $\delta: X \times X \to \reals_{\geq 0}$ satisfies the \emph{$\eta$-approximate triangle inequality} for $\eta \geq 1$ if for all $x, y, z \in X$,
\begin{align*}
    \delta(x, z) \leq \eta \cdot (\delta(x, y) + \delta(y, z)) \,.
\end{align*}

We first generalize \Cref{thm:framework-stocas-advers} to the case where $\delta$ satisfies the approximate triangle inequality.

\begin{theorem}\label{thm:framework-stocas-advers-approx-tri}
    Suppose that the cost function $\delta$ satisfies the $\eta$-approximate triangle inequality for $\eta \geq 1$.
    Then, for every distribution $\mathbb{D}$, given an algorithm $\calA$ for online metric matching, there exists an algorithm $\calA'$ such that $\costD(\calA'; S, n) \leq \eta \cdot (\optD(S, n) + \costD(\calA; n))$ for every set of servers $S$ with $|S| \geq n$.
    Moreover, for balanced markets, if $\calA$ is $\alpha$-competitive in the stochastic setting, then $\calA'$ is $(2\alpha \eta^2 + \eta)$-competitive in the semi-stochastic setting.
\end{theorem}

\Cref{thm:framework-stocas-advers-approx-tri} is proved by replacing every usage of the triangle inequality by the approximate triangle inequality in the proofs of Theorems~\ref{thm:alg-with-addi-info} and~\ref{thm:framework-stocas-advers}, and hence the proof is omitted.

We focus on the case where $X = [0, 1]^d$ and $\delta(x, y) = \norm{x - y}_2^p$ for $p \geq 1$, which we prove that satisfies the $2^{p - 1}$-approximate triangle inequality.
More generally, the following stronger statement holds: For $p \geq 1$, the cost function $\delta(x, y) := \delta'(x, y)^p$, where $\delta'$ satisfies the triangle inequality, satisfies the $2^{p - 1}$-approximate triangle inequality.

\begin{lemma}\label{lem:appro-tri-ineq}
    Suppose that $\delta: X \times X \to \reals_{\geq 0}$ satisfies the triangle inequality.
    Then, the cost function $\delta(x, y) := \delta'(x, y)^p$ for $p \geq 1$ satisfies the $2^{p - 1}$-approximate triangle inequality.
\end{lemma}

\begin{proof}
    For all $x, y, z \in X$,
    \begin{align*}
        \delta(x, z)
        &= \delta'(x, z)^p
        \leq (\delta'(x, y) + \delta'(y, z))^p
        = \left( \frac{\delta'(x, y)}{2} + \frac{\delta'(y, z)}{2} \right)^p \cdot 2^p\\
        &\leq \left( \frac{\delta'(x, y)^p}{2} + \frac{\delta'(y, z)^p}{2} \right) \cdot 2^p
        = (\delta(x, y) + \delta(y, z)) \cdot 2^{p - 1}\,,
    \end{align*}
    where the first inequality holds by the triangle inequality, and the second inequality holds by the convexity of $f(x) = x^p$ for $p \geq 1$.
\end{proof}

Assuming $\mathbb{D}$ to be the uniform distribution over $[0, 1]^d$, we give an $O(1)$-competitive algorithm for $d \geq 1$ and $p \in [1, \max\{d, 2\})$ in the semi-stochastic setting.

\begin{theorem}\label{thm:cost-approx-tri-cr}
    Suppose that $m = n$ and $\mathbb{D}$ is the uniform distribution over $[0, 1]^d$.
    Let $\delta(x, y) = \norm{x - y}_2^p$ for $p \in [1, \max\{d, 2\})$.
    Then, for $d \geq 1$, there exists an $O(1)$-competitive algorithm in the semi-stochastic setting.
\end{theorem}

\subsection{Proof of \Cref{thm:cost-approx-tri-cr}}

We first recall the following estimate of $\optD(n)$.

\begin{lemma}[\cite{ledoux2019optimal}]\label{lem:opt-euc-uni-p}
    Let $\bbD$ be the uniform distribution over $[0, 1]^d$, and let $\delta(x, y) = \norm{x - y}_2^p$.
    Then, for $d, p \geq 1$,
    \begin{align*}
        \optD(n) =
        \begin{cases}
            \Theta (n^{1 - \frac{p}{2}}) & d = 1\\
            \Theta ( n^{1 - \frac{p}{2}}(\log n)^{\frac{p}{2}} ) & d = 2\\
            \Theta(n^{1 - \frac{p}{d}}) & d \geq 3
        \end{cases}\,.
    \end{align*}
\end{lemma}

Now, we are ready to prove \Cref{thm:cost-approx-tri-cr}.
    Let $\calA$ be the algorithm as per \Cref{lem:matching-cost-soar}, and we analyze $\costD(\calA; n)$ for $d = 1$, $d = 2$, and $d \geq 3$ separately by applying \Cref{lem:opt-euc-uni-p}.
    For $d = 1$ and $p \in [1, 2)$,
    \begin{align*}
        \costD(\calA; n)
        = \sum_{t=1}^n \frac{1}{t} \cdot \optD(t)
        = O \left( \sum_{t=1}^n t^{-\frac{p}{2}} \right)
        = O \left( \int_1^n x^{-\frac{p}{2}} \mathrm{d}x \right)
        = O\left( n^{1 - \frac{p}{2}} \right)
        = O\left( \optD(n) \right)\,.
    \end{align*}
    For $d = 2$ and $p \in [1, 2)$,
    \begin{align*}
        \costD(\calA; n)
        &= \sum_{t=1}^n \frac{1}{t} \cdot \optD(t)
        = O \left( \sum_{t=1}^n t^{-\frac{p}{2}} (\log t)^{\frac{p}{2}} \right)
        = O \left( (\log n)^{\frac{p}{2}} \int_1^n x^{-\frac{p}{2}} \mathrm{d}x \right)\\
        &= O \left( n^{1 - \frac{p}{2}} (\log n)^{\frac{p}{2}} \right)
        = O \left( \optD(n) \right)\,.
    \end{align*}
    For $d \geq 3$ and $p \in [1, d)$,
    \begin{align*}
        \costD(\calA; n)
        = \sum_{t=1}^n \frac{1}{t} \cdot \optD(t)
        = O \left( \sum_{t=1}^n t^{-\frac{p}{d}} \right)
        = O \left( \int_1^n x^{-\frac{p}{d}} \mathrm{d}x \right)
        = O \left( n^{1 - \frac{p}{d}} \right)
        = O \left( \optD(n) \right)\,.
    \end{align*}
    As a result, $\calA$ is $O(1)$-competitive in the stochastic setting.
    By \Cref{lem:appro-tri-ineq}, $\delta$ satisfies the $2^{p - 1}$-approximate triangle inequality.
    Combining with \Cref{thm:framework-stocas-advers-approx-tri} leads to an $O(1)$-competitive algorithm in the semi-stochastic setting.

\end{document}